\PassOptionsToPackage{svgnames}{xcolor}
\documentclass[11pt,reqno]{article}
\usepackage{amsmath,hyperref,amssymb, amsthm}
\usepackage{authblk}
\usepackage{graphicx,url}
\usepackage{color}
\usepackage{caption}
\usepackage[backend=biber,style=ieee]{biblatex}
\newcommand{\be}{\begin{equation}}
\newcommand{\ee}{\end{equation}}
\newcommand{\beq}{\begin{equation}}
\newcommand{\eeq}{\end{equation}}

\allowdisplaybreaks \numberwithin{equation}{section}

\newcommand{\CA}{{\cal A}}
\newcommand{\CB}{{\cal B}}

\newcommand{\CH}{\mathcal{H}}

\newcommand{\CI}{{\cal I}}
\newcommand{\CL}{{\cal L}}
\newcommand{\CN}{\mathcal{N}}

\newcommand{\calC}{\mathcal{C}}

\newcommand{\CV}{\mathcal{V}}

\DeclareMathOperator{\Tr}{Tr}

\DeclareMathOperator{\Aut}{Aut}
\DeclareMathOperator{\Com}{Com}

\newcommand{\ZZ}{\mathbb{Z}}
\newcommand{\RR}{\mathbb{R}}
\newcommand{\CC}{\mathbb{C}}
\newcommand{\QQ}{\mathbb{Q}}

\newcommand{\MM}{\mathbb{M}}

\newcommand{\paraf}{\mathcal{P}}

\newcommand{\nord}[1]{:\!\! #1\!\!:}

\newcommand{\normord}[1]{:\mathrel{ #1}:}

\newcommand{\alg}[1]{\mathfrak{#1}}

\newtheorem{theorem}{Theorem}

\DeclareMathOperator{\ch}{ch}

\DeclareMathOperator{\Hom}{Hom}
\DeclareMathOperator{\End}{End}
\DeclareMathOperator{\Rep}{Rep}

\title{Monstrous parafermionic defects and other non-invertible symmetries in chiral CFTs} 
\author{Roberto Volpato\thanks{volpato@pd.infn.it}}
\affil{\small Dipartimento di Fisica e Astronomia `Galileo Galilei', Universit\`a di Padova \\ INFN, sez. di Padova\\ Via Marzolo 8, 35131, Padova, Italy}

\begin{document}

\maketitle

\begin{abstract}
	The structure of non-invertible symmetries in 2D CFT is closely tied to the algebra of preserved holomorphic fields, its representation theory, and its embedding in the full chiral algebra of the CFT. Unfortunately, finding new embeddings of a chiral algebra into another is a difficult mathematical problem, and one of the major obstacles in the pursue of new and exotic categories of non-invertible symmetries.\\
	In this work, we describe a simple general technique that leads to a number of new non-trivial results in this direction: (1) For the holomorphic Monster CFT $V^\natural$, we show that for every `Fricke' non-anomalous Monster element of order $N$, there is an embedding of the parafermion algebra $\frac{su(2)_N}{u(1)}$ in $V^\natural$, we compute the characters of its commutant, and describe the topological defects preserving these subalgebras. (2) We provide an easy-to-check sufficient condition for a holomorphic VOA $V$ to be self-orbifold under a cyclic group of (invertible) symmetries, and determine the associated duality defect. (3) We find several new topological defects in CFTs arising from heterotic strings on $T^4$. (4) We prove a number of  VOA embeddings and topological defects in the Leech lattice CFT, in Schellekens theories, and other CFTs of various central charges, and suggest several generalizations of our methods.
\end{abstract}

\newpage
\tableofcontents

\section{Introduction}

The study of generalized symmetries in quantum field theory is one of the most important developments in theoretical physics in the last ten years (see \cite{McGreevy:2022oyu,Freed:2022iao,Shao:2023gho,Brennan:2023mmt,Bhardwaj:2023kri,Gomes:2023ahz,Schafer-Nameki:2023jdn,Luo:2023ive,Kaidi:2026urc} for recent reviews and references). In the context of two-dimensional conformal field theory (CFT), non-invertible (or categorical) symmetries are described by topological line defects, which structure has been studied for a long time \cite{Verlinde:1988sn,Petkova:2000ip,Fuchs:2002cm,Fuchs:2003id,Fuchs:2004dz,Frohlich:2004ef,Fuchs:2004xi,Frohlich:2006ch,Frohlich:2009gb,Chang_2019}. While ordinary global symmetries are formally described as elements in a group, topological defects should be interpreted as objects in a fusion (or, more generally, tensor) category. It is clearly very important to classify the possible consistent fusion categories of topological defects, and to understand whether all of them appear as generalized symmetries in some CFT.

For rational CFTs, the appearance of some fusion category of topological defect lines is directly related to subalgebras of the (anti-)chiral algebra of (anti-)holomorphic fields \cite{Rayhaun:2023pgc,Gannon:2026ttf}. Specifically, the vector space of local operators that are preserved by all topological defects in a given fusion category (i.e., such that moving any defect across the local operator does not change  any correlation function) is closed with respect to OPE on the sphere. In particular, the preserved (anti-)holomorphic fields form a subalgebra of the (anti-)chiral algebra. Vice versa, given (conformally embedded) rational subalgebras $\CB$, $\bar\CB$ of, respectively, the chiral and antichiral algebras $\CA$ and $\bar\CA$, one can consider the fusion category of topological defects preserving $\CB\otimes \bar\CB$. This connection works at the level of the corresponding abstract categories: roughly speaking, the modular tensor category of representations of the preserved (anti-)chiral subalgebra can be identified with the Drinfeld center of the fusion category of defects; vice versa, the set of primary operators in a given rational CFT and their fusion is specified by an algebra object in the modular tensor category of representations of the chiral algebra, and the fusion category of topological defects is obtained as a category of (non-local) modules for this algebra object -- we refer to \cite{Moller:2024xtt,Rayhaun:2023pgc,Gannon:2026ttf} for more details and precise definitions.

This connection provides a very effective and systematic method to determine the topological defects in a given rational CFT, once it is known which subalgebras can be conformally embedded in the chiral algebra of the theory. Unfortunately, finding the subalgebras in a given chiral algebra is a difficult mathematical problem; we know only a few systematic methods (orbifolds, embeddings of affine algebras, cosets or commutants), that have been exploited for many years. In principle, one could exploit the connection in the opposite direction: if one knows that a category of topological defects appears in a CFT, one could determine new embedded subalgebras by restricting to the preserved holomorphic fields.  In practice, we have essentially no independent way of determining new fusion categories of topological defects in a given CFT (besides the group-like categories of invertible symmetries), unless we know already something about the preserved subalgebra. The situation gets even worse in the non-rational case, where the topological defects can generate tensor categories with infinitely many simple objects; some results have been obtained only for torus models or using constraints from supersymmetry \cite{Angius:2024evd,Angius:2025zlm,Angius:2025ium}.

In this article, we describe a simple method to determine a number of new embeddings of chiral algebras into one another, and then use them to study the associated category of topological defects. Formally, we just use  known techniques in vertex operator algebras (simple current extensions, affine embeddings, orbifolds and commutants, i.e. cosets), but they are combined in a way that is particularly effective and that has not been fully explored so far.

Rather than trying to obtain the most general results, we prefer to describe a number of examples that illustrate the potentialities of our approach. We mostly consider CFTs that are purely holomorphic, so that the whole theory is given by the chiral algebra of holomorphic fields. The most famous example is Frenkel-Lepowsky-Meurman Monster module $V^\natural$ \cite{FLM1988}, that plays a crucial role in the Monster moonshine conjecture \cite{ConwayNorton1979} and in its proof by Borcherds \cite{Borcherds1992}. This CFT is a particularly nice laboratory to test our methods, because the absence of weight $1$ fields prevents the application of standard constructions based on embeddings of affine algebras of currents. The topological defects in this theory have been explored in a number of recent works \cite{Lin:2019hks,Bae:2020pvv,Fosbinder-Elkins:2024hff,Volpato:2024goy,Honda:2026bjy}. In particular, in \cite{Honda:2026bjy}, it has been conjectured that the parafermion algebras $\paraf(N)=\frac{\alg{su}(2)_N}{\alg{u}(1)}$ can be embedded in $V^\natural$ for every $N$ being an odd prime dividing the order of the Monster group $\mathbb{M}\cong \Aut(V^\natural)$. Based on this hypothesis, the authors of \cite{Honda:2026bjy} manage to build a number of new `Monster parafermionic' conformal field theories. In our work, we provide a uniform \emph{proof} that embeddings of $\paraf(N)$ in $V^\natural$ exist for a larger set of values $N$, corresponding to the orders of all `Fricke' elements of the Monster group $\mathbb{M}$ (see section \ref{s:Fricke} for the precise definitions). This result about the FLM Monster module is the consequence of more general theorems in section \ref{s:mainth}, that we later apply (see section \ref{s:examples}) to a number of CFTs -- the lattice VOA based on the Leech lattice, some other Schellekens holomorphic CFTs of central charge $24$ \cite{Schellekens1993}, a VOA of central charge $32$, and certain classes of non-holomorphic CFTs obtained from compactification of heterotic strings on $T^4$.

In section \ref{s:mainth}, we also give a systematic way to determine the $q$-expansion of the characters of the commutant (or coset) $\calC_N=\Com(\paraf(N),V)$ of  $\paraf(N)$ in $V$, i.e. the subalgebra generated by all operators that have non-singular OPE with the operators in $\paraf(N)$. The embedding of the parafermion algebras $\paraf(N)$ implies the existence of a fusion category of topological defects in the CFT; we discuss how these generalized symmetries act on the states of the theory, and how to explicitly compute the twining partition functions. Our final results concern some necessary conditions for a VOA $V$ to be self-orbifold under a (invertible) symmetry $g$ of order $N$, i.e. such that $V$ is isomorphic to the orbifold $V/\langle g\rangle$. This property is closely related to the existence of a topological `duality' defect in the VOA, that extends the $\ZZ_N$ group generated by $g$ to a Tambara-Yamagami (TY) category of generalized symmetries \cite{Tambara:1998}. We discuss such duality defects and the properties of the corresponding TY category.

The techniques we use to prove the general results in section \ref{s:mainth} and that we apply to the examples of section \ref{s:examples} admit a broad range of variations and generalizations, that can be applied in different contexts. We give a flavour of these extensions in a few examples in section \ref{s:extensions}. To make the article self-contained, in the appendices we review the definitions and some known properties of vertex operator algebras and topological defects in CFTs.

\section{Parafermion embeddings and self-orbifold properties }\label{s:mainth}

\subsection{Prelude: a string-theoretical construction}

In the next sections, we will formulate our main results in the language of vertex operator algebras, that provides a rigorous mathematical framework to describe the chiral algebra of a CFT. Nevertheless, the method adopted in the proof of these results is inspired by standard constructions in string theory, so it seems useful to provide an intuition of the main idea in this context. Similar techniques were described in \cite{Persson:2015jka,Persson:2017lkn} in the context of superstring theory, and applied in \cite{Paquette:2016xoo,Paquette:2017xui} to provide a `physics proof' of the genus zero property of Monstrous moonshine.

Let $\calC$ be a (not necessarily holomorphic) 2-dimensional CFT, that we assume to be unitary, with a unique vacuum and compact (i.e. with discrete $L_0$ and $\bar L_0$ spectrum). Let $g$ be an (invertible) symmetry of $\calC$ of finite order $N$, and suppose that it is not anomalous, i.e. the orbifold $\calC/\langle g\rangle$ is a consistent CFT. We would like to know whether the theory is self-orbifold, i.e. if there is an isomorphism $\calC \cong \calC/\langle g\rangle$, and in this case to determine the action of the corresponding duality defect.

To this goal, let us first consider the product $\calC\times S^1$ of $\calC$ times the theory of a free scalar $\phi(z,\bar z)$ compactified on a circle $S^1$ of radius $R$, and then take a `CHL like orbifold' \cite{Chaudhuri:1995fk} by a symmetry $(g,\delta)$ of order $N$, where $g$ acts on the factor $\calC$ and $\delta$ is a shift along $S^1$ by $1/N$ of a period. Notice that the twisted sectors in this orbifold carry some `fractional winding' along $S^1$; the orbifold projection forces the $g$-eigenvalues to match with the momentum quantum number (mod $N$) along $S^1$. It is easy to see that in the limit $R\to \infty$, one simply recovers the product of the original CFT $\calC$ times the uncompactified free scalar $\phi$. In particular,  in this limit the winding states become infinitely massive and disappear from the spectrum, while the orbifold projection is `washed out' as the momenta along the circle reach a continuum; see \cite{Persson:2015jka} for more details. In a similar fashion, in the limit $R\to 0$ one obtains the product of the orbifold theory $\calC/\langle g\rangle$ times the uncompactified dual free scalar $\tilde \phi$.

These two limits suggest that the CHL orbifold, in general, does not admit T-duality: the large and the small radius regimes are inequivalent theories, unless the CFT $\calC$ is self-orbifold. In fact, a more careful analysis shows that a version of T-duality (called `Fricke  T-duality' in \cite{Persson:2015jka}) exists if and only if $\calC\cong \calC/\langle g\rangle$. Thus, we translated the question about the CFT $\calC$ being self-orbifold to the existence of a certain T-duality in the CHL model.

Some basic observations about string theory on circle suggest a possible way to prove this T-duality. Suppose that at some particular radius $R=R^*$, the CHL orbifold contains some additional holomorphic weight $1$ field, that enhances the $\alg{u}(1)$ current $i\partial \phi$ to a $\alg{su}(2)$ algebra; typically, such fields will arise from the twisted sectors.  The zero modes of this $\alg{su}(2)$ algebra generate a group $SU(2)$ of symmetries of the theory at this particular radius. This group contains a symmetry that reverses the sign of the holomorphic $\alg{u}(1)$ current $i\partial \phi(z)$, while leaving its anti-holomorphic counterpart $i\bar\partial\phi(\bar z)$ fixed. As a consequence, this symmetry changes the sign of the exactly marginal operator $\nord{\partial \phi(z)\bar\partial\phi(\bar z)}$ deforming the radius of the circle, so that it establishes an equivalence between the theories at large radius $R>R^*$ and the theories at small radius $R<R^*$. This implies that T-duality holds for such a CHL orbifold, and therefore that the original CFT  $\calC$ is self-orbifold. Finally, the $SU(2)$ transformation at the self-dual radius $R^*$ provides the explicit isomorphism $\calC\cong \calC/\langle g\rangle$, from which one can construct the duality defect \cite{Volpato:2024goy}.

The reason why this approach is very effective is because one needs very little information about the CFTs $\calC$ and $\calC/\langle g\rangle$. In fact, the existence of a single twisted sector state with the  correct properties and conformal weight is sufficient to derive a result constraining  all the twisted and untwisted sectors of the theory.

The purely VOA version of this idea, that will be applied to prove the main theorems in the following sections, can be obtained by focusing on the properties of the extended chiral algebra of the CHL model at the self-dual radius (that will be $R^*=\sqrt{2N}$) and its representations \cite{Paquette:2017xui,Volpato:2024goy}. Before the orbifold,   the chiral algebra of the product $S^1\times \calC$ at $R=\sqrt{2N}$ is of the form $\alg{u}(1)_N\otimes V$. Here, $V$ is the chiral algebra (VOA) of the internal CFT $\calC$, and $\alg{u}(1)_N$ is generated by the current $i\partial\phi$ and by the holomorphic vertex operators $\nord{e^{in\sqrt{2N}\phi(z)}}$, $n\in \ZZ$. After taking the CHL orbifold, the untwisted algebra is projected to $\alg{u}(1)_N\otimes V^g$, but it can get  extended by the holomorphic fields from the twisted sectors to some larger chiral algebra $\tilde V$. As will be shown in the theorems \ref{th:simplecurrent} and \ref{th:paraf}, under some conditions that are quite easy to check, the chiral algebra $\tilde V$ contains an extension of $\alg{u}(1)_N$ to an affine current algebra $\alg{su}(2)_N$. Not only this implies the existence of T-duality (i.e. the fact that $\calC$ is self-orbifold), but also that the $g$-invariant chiral algebra $V^g$ of the internal CFT $\calC$ contains a parafermion subalgebra $\paraf(N)=\frac{\alg{su}(2)_N}{\alg{u}(1)}$. The embedding of a parafermion algebra is related to the existence of a fusion category of generalized symmetries in $\calC$ (including, but not only,  the duality defect for $\langle g\rangle\cong \ZZ_N$). 

While the VOA description is less intuitive (at least for string theorists), it can be made mathematically rigorous, and admits many more generalizations.

\subsection{The main theorems and their proofs}\label{s:proofs}

We will use the concept of vertex operator algebra (VOA) to describe the algebra generated by holomorphic fields in a CFT; see appendix \ref{a:VOA} for some definitions and references. The CFTs we consider are always unitary, compact (discrete $L_0,\bar L_0$ spectrum), and with a unique vacuum. Most of the VOAs $V$ we will consider are \emph{strongly rational}, i.e. rational (finitely many irreducible $V$-modules), simple ($V$ is irreducible, as a module over itself), self-contragredient or self-dual ($V\cong V^*$ as a $V$-module), of CFT-type ($L_0$-eigenvalues are non-negative, and only the vacuum has $L_0=0$), and $C_2$-cofinite (see \cite{Zhu1996ModularInvariance} for the definition).  For $V$ strongly rational, its category of ordinary modules is a modular tensor category \cite{Huang2008Rigidity}, so that in particular there are modular $S$- and $T$-matrices,  and Verlinde formula holds. We will also assume that the conformal weight of irreducible modules different from the vacuum are strictly positive -- this is expected in unitary CFTs.

Let $V$ be a strongly rational bosonic VOA, let $\Rep(V)$ be its modular tensor category of modules, and $Irr(V)=\{M_a\}_{a\in A}$ the finite set of irreducible modules of $V$ (i.e. the simple objects in $\Rep(V)$), with $M_0$ the vacuum representation. Recall that an irreducible module $J\in Irr(V)$ is called a simple current for $V$ if the fusion product with any $M_a\in Irr(V)$ is still irreducible
\be J\boxtimes M_a\cong M_{Ja}\in Irr(V)\ .
\ee Thus, a simple current necessarily defines a cyclic permutation $M_a\mapsto M_{Ja}$ of finite order $N$ on the set $Irr(V)$ of irreducible modules; the order $N$ of $J$ is the smallest positive integer such that $J^N\cong M_0$, where $J^N=J\boxtimes \ldots \boxtimes J$ ($N$ times).  When a simple current $J$ has integral conformal weight $h(J)\in \ZZ$, then there is an extension $\tilde V\supset V$ such that $\tilde V=\oplus_{n=1}^N J^n$ a a $V$-module.
See \cite{Carnahan:2014emx,Creutzig:2015buk} for recent results about simple current extensions. Notice that the contragredient (dual) module of $J$ is $J^{N-1}$.

\begin{theorem}\label{th:simplecurrent}
	Let $V$ be a strongly rational VOA admitting a simple current $J$ of order $N$.  Suppose that $J$ contains a vector $v\neq 0$ such that:
	\begin{enumerate}
		\item the conformal weight of $v$ is $h(v)=1-\frac{1}{N}$;
		\item either $V$ contains no states of weight $1$ ($V_1=0$), or $v$ is in the vacuum representation of the affine Kac-Moody algebra $\alg{h}$ generated by the operators in $V_1$;
	\end{enumerate}
then, there is an embedding of the $\ZZ_N$ parafermion algebra $\paraf(N)=\frac{\alg{su}(2)_N}{\alg{u}(1)_N}$ in $V$.
\end{theorem}
\begin{proof}
	Consider the $c=1$ lattice VOA $W_0:=V_{\sqrt{2N}\ZZ}\cong \alg{u}(1)_N$ of one chiral free boson on the lattice $\sqrt{2N}\ZZ$. It contains a $\alg{u}(1)$ current $j^0(z)$ that we normalize so that $j^0(z)j^0(0)=\frac{1}{z^2}+\ldots$ as well as $\alg{u}(1)$ primary states $\CV_\lambda$ of $u(1)$-charge $\lambda\in \sqrt{2N}\ZZ$. The VOA $W_0$ has $2N$ irreducible modules $W_m$, $m\in \ZZ/2N\ZZ$, containing states $\CV_\lambda$ of $u(1)$-charge $\lambda\in \frac{m}{\sqrt{2N}}+\sqrt{2N}\ZZ$. The modules $W_m$ form a single orbit with respect to the simple current $W_1$ of order $2N$. 
	We consider the product VOA
	\be W_0\otimes V\ ,
	\ee and extend it by the simple current $W_2\otimes J$, that has order $N$ and integral conformal weight $\frac{2^2}{4N}+1-\frac{1}{N}=1$, to get the VOA
	\be \tilde V=\bigoplus_{n\in \ZZ/N\ZZ} W_{2n}\otimes J^n\ .
	\ee The original VOA $V$ can be identified with the commutant  in $\tilde V$ of the subVOA $\alg{u}(1)\subset W_0$ generated by the $\alg{ u}(1)$ current $j^0(z)$
	\be \Com(\alg{u}(1),{\tilde V})\cong V\ .
	\ee Indeed, such a commutant must have zero $\alg{u}(1)$-charge, which implies that it must be contained in the $W_0\otimes V$ component of $\tilde V$, and it must commute with the stress-tensor $T_{W_0}$ of the $W_0$ factor, which is built as a normal ordered product of the $\alg{u}(1)$ current.\\
	Let us consider the weight $1$ fields in $\tilde V$. Besides the algebra $\alg{u}(1)\oplus \alg{h}$ generated by the weight $1$ operators in $W_0\otimes V$, there are (at least) two more weight one fields $j^+:=\CV_{\frac{+2}{\sqrt{2N}}}\otimes v \in W_2\otimes J$, and $j^-:=\CV_{\frac{-2}{\sqrt{2N}}}\otimes v' \in W_{-2}\otimes J^{N-1}$, where $v'\in J^{N-1}$ is a vector such that $(v,v')=1$ with respect to the natural pairing between $J$ and its contragredient module $J^{N-1}$. Notice that $v'$ can be chosen to be in the vacuum representation of $\alg{h}$ (the dual of the $\alg{h}$-representation where $v$ lives). The OPE of $j^+(z)$ and $j^-(z)$ must be contained in the $W_0\otimes V$ component of $\tilde V$ and the OPE of either $j^+$ or $j^-$ with any current in $V$ must be non-singular, so that $j^+(z)$, $j^-(z)$ and $j^0(z)$ form a closed current algebra. A simple calculation shows that it is a $\alg{su}(2)_N$ affine algebra, where $\sqrt{N}j^0(z)$ is a Cartan generator normalized so that the roots have squared length $2$. Therefore, if we denote by
	\be \mathcal{C}_N:=\Com(\alg{su}(2)_N,\tilde V)\ ,
	\ee the commutant of $\alg{su}(2)_N$ within $\tilde V$, we have a conformal embedding
	\be  \alg{su}(2)_N\otimes\mathcal{C}_N\subset \tilde V\ .
	\ee The algebra $\alg{su}(2)_N$ contains the $\alg {u}(1)$ generated by $j^0(z)$, and the commutant of $\alg{u}(1)$ in $\alg{su}(2)_N\otimes \mathcal{C}_N$ is 
	\be \Com(\alg{u}(1),\alg{su}(2)_N\otimes\mathcal{C}_N)= \frac{\alg{su}(2)_N}{\alg{u}(1)}\otimes \mathcal{C}_N\cong \paraf(N)\otimes \mathcal{C}_N\ .
	\ee On the other hand, this must be contained in the commutant of $\alg{u}(1)$ in the larger VOA $\tilde V$, so that we have a conformal embedding
	\be \paraf(N)\otimes \mathcal{C}_N\subset \Com(\alg{u}(1),{\tilde V})\cong V\ ,
	\ee and we conclude.
\end{proof}

A particularly nice example of simple currents is given by the $g$-twisted sector $V_g$ for a (strongly rational) holomorphic VOA $V$ with an automorphism $g\in \Aut(V)$ of finite order $N$.  Theorem \ref{th:paraf} below shows that very general results can be obtained simply by requiring the conformal weight of the $g$-twisted sector to be strictly less than $1$, and to the $g$-twisted ground state to commute with the $g$-fixed currents in $(V^g)_1$. 

Let us first review some basic facts about orbifolds in VOA. For $V$ strongly rational holomorphic and $g\in \Aut(V)$ of order $N$, it is known \cite{Carnahan:2016guf,vanEkeren:2017scl} that the $g$-invariant subVOA $V^g \subset V$ is also strongly rational and has $N^2$ irreducible ordinary modules $V_{a,b}$, $a,b\in \ZZ/N\ZZ$. Recall that a holomorphic VOA $V$ with cyclic symmetry $\langle g\rangle$  admits a unique (up to isomorphisms) irreducible $g^a$-twisted $V$-module $V_{g^a}$. In particular, the conformal weight of the $g$-twisted sector, for $g$ of order $N$, takes values in
\be \frac{t}{N^2}+\frac{1}{N}\ZZ\ ,
\ee where $t\in \ZZ/N\ZZ$ is the order of the 't Hooft anomaly, i.e. of the $3$-cohomology class $[\alpha]\in H^3(\ZZ_N,U(1))\cong \ZZ_N$ that determines a non-trivial associator for the corresponding fusion category of invertible defects. The orbifold $V/\langle g\rangle$ is a consistent VOA if and only if the anomaly vanishes, i.e. $t=0$.

From now on, let us assume that $g$ is non-anomalous. The action of $g$ on $V$ can be extended to a symmetry of the same order $N$ on each $g^a$-twisted sector; this extension is determined up to an arbitrary $N$-th root of unity. Each ${g^a}$-twisted module $V_{g^a}$ decomposes into ordinary modules $V_{a,b}$ for the $g$-invariant subVOA $V^g$
\be V_{g^a}=\bigoplus_{b\in \ZZ/N\ZZ} V_{a,b}\ .
\ee
In particular, $V_{a,b}$ is the $g=e^{2\pi i\frac{b}{N}}$ eigenspace in $V_{g^a}$:
\be V_{a,b}:=\{ v\in V_{g^a}\mid g(v)=e^{2\pi i \frac{b}{N}}v\}\ ,\qquad a,b\in \ZZ/N\ZZ\ .
\ee
The ambiguity in the choice of the $g$-action on the twisted sectors can be fixed by requiring
the conformal weights of $V_{a,b}$ to be
\beq\label{weights} h_{V_{a,b}}=\frac{ab}{N}\mod \ZZ\ ,
\eeq
and to satisfy the group-like fusion rules
\beq\label{fusion} V_{i,j}\boxtimes V_{k,l}\cong  V_{i+k,j+l}\ , \qquad i,j,k,l\in \ZZ/N\ZZ\ .
\eeq
In this convention, the action of $g$ on the $g^1$-twisted sector $V_g$ is defined by
\be g_{\rvert V_{g}}=e^{2\pi i L_0}\ ,
\ee while on a generic $g^a$-twisted sector it is defined in such a way that \eqref{fusion} hold.
The orbifold theory 
\be
V/\langle g\rangle:=\bigoplus_{a\in \ZZ/N\ZZ} V_{a,0}\ ,
\ee is a holomorphic VOA with the same central charge.

\begin{theorem}\label{th:paraf}
	Let $V$ be a strongly rational holomorphic VOA, and $g\in \Aut(V)$ be a non-anomalous automorphism of finite order $N$. Let $(V^g)_1$ the space of currents (vectors of conformal weight $1$) of  the $g$-fixed subVOA $V^g\subset V$. If the $g$-twisted sector $V_g$ contains a non-zero vector $v\in V_g$ such that:
	\begin{enumerate}
		\item the conformal weight $h(v)$ satisfies
		$$ 0<h(v)<1
		$$ and
		\item either $(V^g)_1=0$, or $v$ is in the vacuum representation of the current algebra $\alg{h}$ generated by the weight $1$ operators in $(V^g)_1$,
	\end{enumerate} then:
	\begin{enumerate}
		\item The conformal weight of $v$ is necessarily $h(v)=1-1/N$;
		\item $V$ is self-orbifold with respect to $g$, i.e. the orbifold VOA $V/\langle g\rangle$ is isomorphic to $V$;
		\item There is an embedding of the $\ZZ_N$ parafermion algebra $\paraf(N)$ in $V^g$. 
	\end{enumerate} 
\end{theorem}
\begin{proof}
Because $g$ is not anomalous, the conformal weight of the $g$-twisted sector is in $\frac{1}{N}\ZZ$, and the condition $0<h(v)<1$ implies that $h(v)=1-\frac{r}{N}$ for some integral $r$ with $0<r<N$. Using the decomposition $V_g=\oplus_{b\in \ZZ/N\ZZ} V_{1,b}$ into $g$-eigenspaces, and the fact that $V_{1,b}$ has conformal weights in $\frac{b}{N}+\ZZ$, we conclude that $v$ must be contained in the component $V_{1,-r}$. Furthermore, a state $v'$ in the vacuum representation of the current algebra $\alg{h}$ and with the same conformal weight must be contained in the dual (contragredient) module $V_{-1,r}\subset V_{g^{-1}}$. Let us show that $r$ must be $1$. We  use an analogous construction as in the proof ot theorem \ref{th:simplecurrent}, but now based on the free boson VOA $W_0:=V_{\sqrt{2rN}\ZZ}$ based on the lattice $\sqrt{2rN}\ZZ$. We take the product $W_0\otimes V^g$, and denote by
\be V_{m,a,b}:=W_m\otimes V_{a,b}\ ,\qquad m\in \ZZ/2rN\ZZ,\quad a,b\in \ZZ/N\ZZ\ ,
\ee the modules of this product algebra. Then, we can extend $V_{0,0,0}\cong W_0\otimes V^g$ by the simple current $V_{2r,1,-r}=W_{2r}\otimes V_{1,-r}$ of order $N$, that has a ground state $\CV_{\frac{+2r}{\sqrt{2rN}}}\otimes v$ of conformal weight $1-\frac{r}{N}+\frac{(2r)^2}{4rN}=1$, to get the VOA
\be \tilde V=\oplus_{n\in \ZZ/N\ZZ} V_{2rn, n,-nr}\ .
\ee The modules of $\tilde V$ decompose into modules $V_{m,a,b}$ of the subVOA $V_{0,0,0}$ that are local with respect to the simple current $V_{2r,1,-r}$, i.e. such that the difference of conformal weights $h(V_{m+2r,a+1,b-r})-h(V_{m,a,b})$ is integral. The latter condition translates into
\begin{align} \frac{(m+2r)^2}{4Nr}+\frac{(a+1)(b-r)}{N}- \frac{m^2}{4Nr}-\frac{ab}{N}\in \ZZ \quad &\Leftrightarrow \quad \frac{4r(m+r)}{4Nr}+\frac{b-r-ar}{N}\in \ZZ  \\ &\Leftrightarrow \quad  m-ar+b\equiv 0\mod N\ZZ\ .
\end{align} The modules of $\tilde V$ are therefore given by
\be \tilde V_t=\bigoplus_{n\in \ZZ/N\ZZ} V_{2nr-t,n,t-nj}\qquad t\in \ZZ/2N\ZZ\ .
\ee By the same argument as in the proof of theorem \ref{th:simplecurrent}, the VOA $\tilde V$ contains an affine subalgebra $\alg{su}(2)$ at some level. Therefore, the group of inner automorphisms of $\tilde V$ contains a $SO(3)$ subgroup generated by the zero modes of these $\alg{su}(2)$ currents. In particular, there is an involution $h$ whose adjoint action on the three $\alg{su}(2)$ currents  maps $j^0(z)$ to $-j^0(z)$. Being an inner automorphism, it must map each module $\tilde V_t$ into itself, and it must flip the sign of the $j^0$-charge. This implies that in each module $\tilde V_t$, the set of $j^0$-charges must be symmetric with respect to change of sign. But if we take, for example, the module
\be \tilde V_1=\bigoplus_{n\in \ZZ/N\ZZ} V_{2nr-1,n,1-nr}\ ,
\ee we see that the $j^0$-charges take values in $\frac{1}{\sqrt{2rN}} (-1+2r\ZZ)$. This set is symmetric with respect to $0$ if and only if $r=1$, so this is the only consistent value of $r$, and point (1) of the theorem is proved.\\
Let us now set $r=1$, and prove that the theory is self-orbifold, i.e $V\cong V/\langle g\rangle$. The inner automorphism $h$ maps $j^0$ to $-j^0$, so it must map the $0$-charge component $V_{0,0,0}=W_0\otimes V^g$ of $\tilde V$ to itself, and it must be given by the tensor product $h=f_{W}\otimes f_{V^g}$ of the charge conjugation automorphism $f_W$ in $W_0$ and an involution $f_{V^g}$ of  $\Com(\alg{u}(1),\tilde V)\cong V^g$. As mentioned above, because it is inner it must map each $\tilde V$-module 
\be \tilde V_t=\bigoplus_{n\in \ZZ/N\ZZ} V_{2n-t,n,t-n}=\bigoplus_{\substack{m\in \ZZ/2N\ZZ\\m\equiv t\mod 2}} V_{m,\frac{m+t}{2},\frac{t-m}{2}}=\bigoplus_{\substack{m\in \ZZ/2N\ZZ\\m\equiv t\mod 2}} W_m\otimes V_{\frac{m+t}{2},\frac{t-m}{2}}
\ee  into itself, and map each component $ V_{m,\frac{m+t}{2},\frac{t-m}{2}}$ to the component $V_{-m,\frac{-m+t}{2},\frac{t+m}{2}}$ with opposite $j^0$ charge. As a consequence, the involution $f_{V^g}$ must map each $V^g$-module $V_{a,b}$ to the module $V_{b,a}$, for all $a,b\in \ZZ/N\ZZ$. Because $V$ and $V/\langle g\rangle$ are given by
\be V=\bigoplus_{n\in \ZZ/N\ZZ} V_{0,n}\ ,\qquad V/\langle g\rangle=\bigoplus_{n\in \ZZ/N\ZZ} V_{n,0}
\ee it follows that the involution $f_{V^g}$ extends to a VOA isomorphism $f: V\to V/\langle g\rangle$, thus proving point (2).
\\
Finally, to prove point (3) we just notice that, because $r=1$, then the hypotheses of theorem \ref{th:simplecurrent} are all satisfied, so that there is an embedding of the parafermion algebra $\paraf(N)$ in $V^g$.
\end{proof}

In the proof of this theorem, and in particular the statement that the conformal weight is necessarily $1-1/N$, it is crucial that the $g$-twisted state $v$ is in the vacuum representation of the $g$-fixed affine algebra. If we drop this hypothesis, there are certainly many examples where the conformal weight is $1-r/N$ with $r>1$. We will discuss some examples in the following sections.


\subsection{Characters and representations of the commutant algebra}\label{s:commutant}

Let us discuss the properties of the commutant $\mathcal{C}_N$ of $\alg{su}(2)_N$ in the VOA $\tilde V$ appearing in the proofs of theorems \ref{th:simplecurrent} or \ref{th:paraf}; equivalently, $\mathcal{C}_N$ is the commutant of $\paraf(N)$ in $V$ or $V^g$, respectively. First, by the standard coset construction, the difference
\be T_\mathcal{C}:=T_{\tilde V}-T_{\mathfrak{su}(2)_N}\ ,
\ee between the stress-energy tensor $T_{\tilde V}$ of $\tilde V$ and the Sugawara stress-tensor $T_{\mathfrak{su}(2)_N}$ of $\alg{su}(2)_N$ commutes with all currents in $\mathfrak{su}(2)_N$, and therefore provides a well-defined stress tensor for the commutant $\mathcal{C}$, with central charge $c_\mathcal{C}= c_{\tilde V}-\frac{3N}{N+2}$.

There are two subtle questions about $\calC_N$ that we are not able to answer in general in this work: (1) whether $\calC_N$ is rational, and (2) whether $\paraf(N)$ and $\calC_N$ form a \emph{dual pair} of commuting subalgebras, namely if it is also true that the commutant of $\calC_N$ in $V$ or $V^g$ is exactly $\paraf(N)$, and not a non-trivial extension of it. As discussed below, the second question can be answered in particular cases, using some of the tools described in this section.

In the following, we focus on the case of theorem \ref{th:paraf}, where $V$ is  a strongly rational holomorphic VOA with an automorphism $g$ of finite order $N$, and the $g$-twisted sector contains a state of conformal weight $1-1/N$ in the trivial representation of the $g$-invariant current algebra $(V^g)_1$. In the rest of this section, we show how to explicitly compute the characters of the subVOA $\mathcal{C}_N=\Com(\paraf(N),V^g)$ commuting with the embedded parafermion $\paraf(N)=\frac{\alg{su}(2)_N}{\alg{u}(1)_N}$, and of (some of) its modules.

Recall that the irreducible modules $\paraf(N,[l,m])$ of $\paraf(N)$ are labeled by an $\alg{su}(2)_N$ index $l\in\{0,1,\ldots,N\}$ and a $\alg{u}(1)_N$ index $m\in \ZZ/2N\ZZ$, subject to the condition
\be l-m\in 2\ZZ\ ,
\ee and to the field identification
\be \paraf(N,[l,m])\cong \paraf(N,[N-l,m\pm N])
\ .
\ee In the case where $\paraf(N)$ and $\calC_N$ are a dual pair of strongly rational commuting subalgebras in the holomorphic VOA $V$, we expect the modular tensor category $\mathrm{Rep}(\calC_N)$ of $\calC_N$ modules to be the opposite of the category of $\mathrm{Rep}(\paraf(N))$, i.e. to have the same objects and tensor structure, just with the braiding and twisting inverted. Thus, we expect irreducible $\calC_N$-modules $\calC(N,{[l,m]})$ to carry the same labels as the $\paraf(N)$-modules, with the same fusion rules, and with $T$ and $S$-matrices being the complex conjugate of the ones of $\paraf(N)$. This means that the modular data associated with the product $\paraf(N)\otimes \calC_N$ is the same as the product $\paraf(N)\otimes \overline{\paraf(N)}$ of a chiral and antichiral copy of the parafermion algebra in a non-holomorphic CFT, and the classification of modular invariants is exactly the same. Even without proving rationality of $\calC_N$, one can expect the VOA $V$ to decompose into a finite number of modules $\paraf(N,[l,m])\otimes \calC(N,{[l,m]})$ for $\paraf(N)\otimes \calC_N$. Furthermore, the category of $\calC_N$-modules given by finite direct sums of the $\calC(N,{[l,m]})$ appearing in this decomposition (that might not contain all possible ordinary modules for $\calC_N$), should be a tensor category with some `nice' properties; see for example \cite{McRae2024MirrorEquivalence} for recent results in this sense. Our goal in this section is to provide an effective algorithm to obtain the $q$-expansions of the characters $f_{[l,m]}$ of the $\calC_N$-modules $\calC(N,{[l,m]})$.

We assume that the twisted-twining partition functions for $g$ are known
\be T_{g^a,g^k}(\tau):=\Tr_{V_{g^a}}(g^k\,q^{L_0-\frac{c}{24}})\ .
\ee Then, the characters $\tilde T_{a,b}$ of the $V^g$-modules $V_{a,b}$ are given by
\be \tilde T_{a,b}(\tau)=\frac{1}{N}\sum_{k\in \ZZ/N\ZZ} e^{-2\pi i\frac{kb}{N}}T_{g^a,g^k}(\tau)\ .
\ee
In terms of these functions, we can easily derive the characters of the $\tilde V$-modules $\tilde V_t$
\be\label{tildeVchars} \ch_{\tilde V_t}(\tau,z):=\Tr_{\tilde V_t}(q^{L_0-\frac{c+1}{24}}y^{j_0})=\sum_{n\in \ZZ/N\ZZ} \frac{\Theta^{(N)}_{2n-t}(\tau,z)}{\eta(\tau)}\tilde T_{n,t-n}(\tau)\ ,\qquad y=e^{2\pi i z}\ ,
\ee 
where we kept track of the $\alg{u}_1$ charge using the `flavoured' $\alg{u}(1)_N$ theta series
\be \Theta_m^{(k)}(\tau,z)=\sum_{n\in \ZZ} q^{\frac{1}{4k}(m+2kn)^2}y^{\frac{m+2kn}{2}}\ .\ee
Because the current algebra $\alg{su}(2)_N$ is embedded in $\tilde V$,  each $\ch_{\tilde V_t}(\tau,z)$ must admit a decomposition into characters of $\alg{su}(2)_N$ and of its commutant $\mathcal{C}_N$
\be\label{su2exp} \ch_{\tilde V_t}(\tau,z)=\sum_{\substack{l=0\\ l\equiv t\bmod{2}}}^N\ch^{\alg{su}(2)_N}_l(\tau,z)f_{[l,t]}(\tau)\ .
\ee The restriction on the sum over $l$ is due to the fact that  each $\alg{su}(2)_N$ representation contains either only odd or only even $j^0$-charges, depending on the parity of $l$. One can invert the identity \eqref{su2exp} and use \eqref{tildeVchars} to unambiguously determine all the functions $f_{[l,t]}(\tau)$ up to arbitrary order in $q$. 
Explicitly, using formulas \eqref{su2Nchar} and \eqref{invertsu2N} for the $\ch^{\alg{su}(2)_N}_l$ characters, one obtains
\begin{align} f_{[l,t]}(\tau)&=q^{-\frac{(l+1)^2}{4N+8}}\left[\sum_{n\in \ZZ/N\ZZ} \frac{\Theta^{(N)}_{2n-t}(\tau,z)}{\eta(\tau)}(\Theta^{(2)}_1(\tau,z)-\Theta^{(2)}_{-1}(\tau,z))\tilde T_{n,t-n}(\tau)\right]_{y^{l+1}}\\
	&=\sum_{n\in \ZZ} q^{\frac{N(N+2)}{8}(\frac{2n-t}{N}-\frac{l+1}{N+2})^2} (-1)^{\frac{l+t}{2}-n}\frac{\tilde T_{n,t-n}(\tau)}{\eta(\tau)}	\ ,
\end{align} where $\left[\ldots\right]_{y^{l+1}}$ denotes the coefficient of the $y^{l+1}$ term.

 Notice that the $\tilde V$-modules $\tilde V_t$ and $\tilde V_{t+N}$ are built in terms of the same $V^g$-representations $V_{n,t-n}$, tensored with different $\alg{u}(1)_N$ modules. The net effect is that the character of $\tilde V_{t+N}$ is obtained from the one of $\tilde V_t$ simply by replacing each $\alg{su}(2)_N$ character $\ch^{\alg{su}(2)_N}_l(\tau,z)$ with $\ch^{\alg{su}(2)_N}_{N-l}(\tau,z)$, i.e.
\be\label{su2exp2} \ch_{\tilde V_{t+N}}(\tau,z)=\sum_{\substack{l=0\\ l\equiv t\bmod{2}}}^N\ch^{\alg{su}(2)_N}_{N-l}(\tau,z)f_{[l,t]}(\tau)\ .
\ee  By comparing \eqref{su2exp} and \eqref{su2exp2}, we get
\be f_{[l,t]}(\tau)=f_{[N-l,t+N]}(\tau)\ .
\ee These identities strongly suggest that the $\calC_N$-modules labeled by $[l,t]$ and $[N-l,t+N]$ are isomorphic, in agreement with the suggestion that the category of representation of $\calC_N$ is the opposite as the one of $\paraf(N)$-modules.

By expanding the $\alg{su}(2)_N$ characters in \eqref{su2exp} in terms of characters of the subalgebra $\paraf(N)\otimes \alg{u}(1)_N\subset \alg{su}(2)_N$, we obtain
\be\label{tildeVchars2} \ch_{\tilde V_t}(\tau,z)=\sum_{\substack{l=0\\l\equiv t\bmod{2}}}^N\sum_{\substack{m\in \ZZ/2N\ZZ\\ m\equiv l\bmod{2}}}\frac{\Theta^{(N)}_{m}(\tau,z)}{\eta(\tau)}\chi^{\paraf(N)}_{[l,m]}(\tau)f_{[l,t]}(\tau)\ .
\ee 
By consistency, the $n$-th term in \eqref{tildeVchars} must match with the $m$-th term in \eqref{tildeVchars2}, where $m=2n-t$, and therefore
\be \tilde T_{n,t-n}(\tau)=\sum_{\substack{l=0\\l\equiv t\bmod{2}}}^N \chi^{\paraf(N)}_{[l,2n-t]}(\tau)f_{[l,t]}(\tau)\ ,
\ee or equivalently
\be \tilde T_{a,b}(\tau)=\sum_{\substack{l=0\\l\equiv a+b\bmod{2}}}^N \chi^{\paraf(N)}_{[l,a-b]}(\tau)f_{[l,a+b]}(\tau)\ .
\ee In particular, the partition function $Z_V(\tau)$ of the original VOA $V$ can be expanded as
\be\label{diaginvar} Z_V(\tau)=\sum_{b\in \ZZ/N\ZZ} \tilde T_{0,b}(\tau)=\sum_{b\in \ZZ/N\ZZ} \sum_{\substack{l=0\\l\equiv b\bmod{2}}}^N \chi^{\paraf(N)}_{[l,-b]}(\tau)f_{[l,b]}(\tau)\ ,
\ee which corresponds to the following decomposition of $V$ into $\paraf(N)\otimes \calC_N$ modules
\be V=\bigoplus_{b\in \ZZ/N\ZZ} \bigoplus_{\substack{l=0\\l\equiv b\bmod{2}}}^N \paraf(N,[l,-b])\otimes \calC_N([l,b])\ .
\ee
The sum is over all distinct $\paraf(N)$-modules, each one with multiplicity $1$, tensored with the corresponding module for $\calC_N$. Thus, the VOA $V$ is the holomorphic analogue of the diagonal modular invariant for $\paraf(N)$. 

When the characters $f_{[l,t]}$ are explicitly known, one can determine whether the commutant $\Com(\calC_N,V)$ of $\calC_N$ in $V$ is $\paraf(N)$, i.e. if $\paraf(N)$ and $\calC_N$ are a dual pair, by checking whether the modules $\calC_N([l,b])$ different from the vacuum module $\calC_N([0,0])\cong \calC_N$ have strictly positive conformal weight. 

One can also obtain the partition function of the orbifold VOA $V/\langle g\rangle$ by
\be Z_{V/\langle g\rangle}(\tau)=\sum_{a\in \ZZ/N\ZZ} \tilde T_{a,0}(\tau)=\sum_{a\in \ZZ/N\ZZ} \sum_{\substack{l=0\\l\equiv a\bmod{2}}}^N \chi^{\paraf(N)}_{[l,a]}(\tau)f_{[l,a]}(\tau)
\ee which corresponds to the charge conjugate modular invariant. The fact that $V$ and $V/\langle g\rangle$ correspond to different modular invariants implies that every isomorphism $V\stackrel{\cong}{\longrightarrow} V/\langle g\rangle$ is not the identity when restricted to the common subVOA $\paraf(N)\otimes \calC_N$, but acts by some charge conjugation outer automorphism either on $\paraf(N)$ or on $\calC_N$.

In the next sections, we will provide some examples of these characters.


\subsection{Topological defects preserving the parafermion and its commutant}\label{s:topdefparaf}

The decomposition \eqref{diaginvar} suggests that the theory $V$ admits a fusion category of topological defects isomorphic to the category $\mathrm{Rep}(\paraf(N))$ of modules for the parafermion algebra $\paraf(N)$ (see for example \cite{Haghighat:2023sax}). The definition of such defects is completely analogous to the definition of Verlinde lines in a diagonal modular invariant of a rational VOA, in this case $\paraf(N)$. 

Explicitly, this means that there is one simple defect $\CL_{[l,t]}$ for each simple module $\paraf(N,[l,t])$ in $\mathrm{Rep}(\paraf(N))$. Because $\CL_{[l,t]}$  preserves the subVOA $\paraf(N)\otimes \calC_N$, this defect is completely determined by describing the action of the corresponding linear operator $\hat\CL_{[l,t]}$ on all primary states with respect to this algebra.  By \eqref{diaginvar}, the VOA $V$ contains exactly one $(\paraf(N)\otimes \calC_N)$-primary state $\Phi_{[l',t']}$ for each irreducible $\paraf(N)$-module $\paraf(N,[l',t'])$. All such primary states must be simultaneous eigenstates for all simple defect operators $\hat\CL_{[l,t]}$, and the eigenvalues are constrained by the Cardy condition \cite{Petkova:2000ip,Chang_2019}. The solution to these constraints is essentially unique, namely
\be \hat \CL_{[l,t]}(\Phi_{[l',t']}) =\lambda_{[l,t],[l',t']}\Phi_{[l',t']},\qquad \lambda_{[l,t],[l',t']}:= \frac{S_{[l,t],[l',t']}}{S_{[0,0],[l',t']}}\ ,
\ee where \be S_{[l,t],[l',t']}=\frac{2}{\sqrt{2(N+2)}}e^{\pi i \frac{tt'}{N}}\sin\left(\pi \frac{(l+1)(l'+1)}{N+2}\right) \ee
is the S-matrix of the parafermion algebra $\paraf(N)$. 

The $\CL_{[l,t]}$-twining partition functions can be easily computed using \eqref{diaginvar} as
\be Z^{[l,t]}(\tau):=\Tr_V(\hat\CL_{[l,t]}q^{L_0-\frac{c}{24}})=\sum_{b\in \ZZ/N\ZZ} \sum_{\substack{\ell=0\\\ell\equiv b\bmod{2}}}^N \lambda_{[l,t],[\ell ,b]}\chi^{\paraf(N)}_{[\ell ,b]}(\tau)f_{[\ell,-b]}(\tau)\ .
\ee Similarly, the partition function of the $\CL_{[l,t]}$-twisted sector is given by
\be Z_{[l,t]}(\tau):=\Tr_{V_{\CL_{[l,t]}}}(q^{L_0-\frac{c}{24}})=\sum_{t',t''\in \ZZ/N\ZZ} \sum_{\substack{l'=0\\l'\equiv t'\bmod{2}}}^N \sum_{\substack{l''=0\\l''\equiv t''\bmod{2}}}^N N_{[l,t],[l',t']}^{[l'',t'']}\chi^{\paraf(N)}_{[l',t']}(\tau)f_{[l'',-t'']}(\tau)\ ,
\ee where
\be N_{[l,t],[l',t']}^{[l'',t'']}=\sum_{b\in \ZZ/N\ZZ}\sum_{\substack{\ell=0\\\ell\equiv b\bmod{2}}}^N\frac{S_{[l,t],[\ell,b]}S_{[l',t'],[\ell,b]}S^*_{[l'',t''],[\ell,b]}}{S_{[0,0],[\ell,b]}}
\ee are the fusion coefficients of the $\mathrm{Rep}(\paraf(N))$ category.

When $N$ is odd, the category $\mathrm{Rep}(\paraf(N))$ 
contains  the subcategory $ \mathrm{Rep}(\alg{so}(3)_N)$ of representations of $\alg{so}(3)_N$, generated by $\CL_{[l,0]}$, $l\in \{0,\ldots, N\}$, with $l$ even. Furthermore, for all $N$, $\mathrm{Rep}(\paraf(N))$ contains a $\ZZ_N$ group category generated by $\CL_{[0,2t]}$. This is of course just the symmetry group $\langle g\rangle$, so that we usually just write
\be \CL_{[0,2t]}\equiv \CL_{g^t}\ ,\qquad t\in \ZZ/N\ZZ\ .
\ee

The category $\mathrm{Rep}(\paraf(N))$ does not contain the duality defect related to the fact that the theory $V$ is self-orbifold, $V\cong V/\langle g\rangle$. The reason is that the duality defect does not preserve $\paraf(N)\otimes \calC_N$, but acts by charge conjugation on one of the two factors. It is therefore useful to consider an extended fusion category $\widetilde{\mathrm{Rep}}(\paraf(N))$ that only preserves the subalgebra \be \paraf(N)^C\otimes\calC_N\subset \paraf(N)\otimes\calC_N\ ,\ee
where $\paraf(N)^C$ is the subVOA of $\paraf(N)$ fixed by the charge conjugation. 

In general, a duality defect $\CN$ for the group $\langle g\rangle\cong \ZZ_N$ satisfies
\be \CN \CL_{g^k}=\CL_{g^k}\CN=\CL_{g^k}\ ,\qquad \CN^*=\CN\ ,\qquad \CN^2=\sum_{k\in \ZZ/N\ZZ}\CL_{g^k}\ .
\ee As a consequence of these fusion rules, any duality defect $\CN$ annihilates all primary states that are not $\ZZ_N$ invariant, i.e. the ones with $b\neq 0$ in eq.\eqref{diaginvar}. On the $g$-invariant subalgebra, the associated operator $\hat\CN$  acts by $\sqrt{N}$ times an outer automorphism of order $2$ on $V^g$ that exchanges the modules $V_{0,b}$ with the modules $V_{a,0}$. Vice versa, any such outer automorphism of $V^g$ defines a duality defect $\CN$.

The extended fusion category $\widetilde{\mathrm{Rep}}(\paraf(N))$ contains one such duality defect $\CN$, whose corresponding automorphism on $V^g$ is induced by an inner $SO(3)$ symmetry of the extended VOA $\tilde V$, as in the proof of theorem \ref{th:paraf}. By construction, such a defect preserves the subVOA $\calC_N=\Com(\alg{su}(2)_N,\tilde V)$, while it acts on $\paraf(N)$ by a charge conjugation. 

A technique to compute the twining character of $\CN$ was provided in \cite{Volpato:2024goy}; here,  we summarize and refine this idea. Recall from the proof of theorem \ref{th:paraf} that the product $W_0\otimes V_{0,0}=\alg{u}(1)_N\otimes V_{0,0}$ can be extended to a VOA $\tilde V$ containing a $\alg{su}(2)_N$ algebra, with the current zero modes generating a $SO(3)$ Lie group of inner automorphisms of $\tilde V$. Consider the Weyl reflection of this $\alg{su}(2)_N$ changing the sign of the Cartan $\alg{u}(1)$ subalgebra, and let $f:\tilde V\to \tilde V$ be its lift to the $SO(3)$ group of inner automorphisms.  Notice that, while the lift of the non-trivial Weyl element of the $su(2)$ Lie algebra to the simply connected Lie group $SU(2)$ has order $4$, the group acting faithfully on $\tilde V$ is  $SO(3)\cong SU(2)/\ZZ_2$, so $f$ is an involution. Furthermore, $f$ restricts to an automorphism of the subVOA $W_0\otimes V_{0,0}$ of the form $f_1\otimes f_2$, where $f_1$ is the standard charge conjugation automorphism of $W_0\cong \alg{u}(1)$, while $f_2:V_{0,0}\to V_{0,0}$ is the outer automorphism of $V^g=V_{0,0}$ associated with the duality defect $\CN$. Thus, the twining partition function for $\CN$ is given  by
\be Z^{\CN}(\tau):=\Tr_V(\hat\CN q^{L_0-\frac{c}{24}})=\sqrt{N}\Tr_{V_{0,0}}(f_2q^{L_0-\frac{c}{24}})\ .
\ee Let us consider the $f$-twining partition function for the VOA $\tilde V$
\be \Tr_{\tilde V}(fq^{L_0-\frac{c+1}{24}})=\Tr_{W_0\otimes V_{0,0}}(f_1\otimes f_2q^{L_0-\frac{c+1}{24}})=\Tr_{W_0}(f_1q^{L_0-\frac{1}{24}})\Tr_{ V_{0,0}}(f_2q^{L_0-\frac{c}{24}})\ ,
\ee where the first identity follows because only the states with zero $\alg{u}(1)$ charge contribute to the trace of $f$, and all such states are contained in the subVOA $W_0\otimes V_{0,0}$. Thus, we get
\be  Z^{\CN}(\tau)=\sqrt{N}\frac{\Tr_{\tilde V}(fq^{L_0-\frac{c+1}{24}})}{\Tr_{W_0}(f_1q^{L_0-\frac{1}{24}})}\ .
\ee The computation of the denominator in this equation is standard: only the states with zero $\alg{u}(1)$ charge contribute, and we obtain
\be \Tr_{W_0}(f_1q^{L_0-\frac{1}{24}})=\frac{1}{q^{\frac{1}{24}}\prod_{n=1}^\infty(1+q^n)}=\frac{\eta(\tau)}{\eta(2\tau)}\ .
\ee To compute the numerator, we use the fact that $f$ is conjugate to a symmetry of order $4$ in the Cartan torus of $SU(2)$ (and order $2$ on the $SO(3)$ representations), so that
\be \Tr_{\tilde V}(fq^{L_0-\frac{c+1}{24}})=\sum_{n\in \ZZ/N\ZZ} \frac{\Theta^{(N)}_{2n}(\tau,\frac{1}{4})}{\eta(\tau)}\tilde T_{n,-n}(\tau)=\sum_{\substack{l=0\\ l\text{ even}}}^N\ch^{\alg{su}(2)_N}_{l}(\tau,\tfrac{1}{4})f_{[l,0]}(\tau)\ .
\ee We obtain
\be\label{dualitytwin}  Z^{\CN}(\tau)=\sqrt{N}\sum_{\substack{l=0\\l\text{ even}}}^N\ch^{\alg{su}(2)_N}_{l}(\tau,\tfrac{1}{4})\frac{\eta(2\tau)}{\eta(\tau)}f_{[l,0]}(\tau)\ .
\ee
This formula suggests that the twining partition function of $\CL_{[l,t]}\CN$ is given by
\be  Z^{\CN\CL_{[l,t]}}(\tau)=\sqrt{N}\sum_{\substack{\ell=0\\\ell \text{ even}}}^N\lambda_{[l,t],[\ell ,0]}\ch^{\alg{su}(2)_N}_{\ell}(\tau,\tfrac{1}{4})\frac{\eta(2\tau)}{\eta(\tau)}f_{[\ell,0]}(\tau)\ .
\ee The duality defect $\CN$, together with the invertible $\ZZ_N$ defects $\CL_{g^k}$, generate a Tambara-Yamagami category $TY(\ZZ_N,\chi,\epsilon)$ \cite{Tambara:1998}. For a given abelian group $A\cong \ZZ_N$, such a category is determined by a symmetric non-degenerate bicharacter $\chi: A\times A\to U(1)$, and by a sign $\epsilon\in \{\pm1\} $. The bicharacter $\chi$ determines the precise mapping between $g^a$-twisted sectors $V_{a,0}$ and $g$-eigenspaces $V_{0,b}$, determined by the outer automorphism in $V_{0,0}$. In particular, a TY category for $\langle g\rangle\cong\ZZ_N$  with bicharacter $\chi$, the outer automorphism exchanges the $g^a$-twisted sector $V_{a,0}$ to the $\ZZ_N$-representation where $g^k$ acts by $\chi(g^a,g^k)\in U(1)$. The defect $\CN$ we are interested in maps the module $V_{a,0}$ to $V_{0,a}$, so that the relevant bicharacter is
\be \chi(g^a,g^k)=e^{2\pi i\frac{ak}{N}}\ .
\ee

\section{Examples}\label{s:examples}

\subsection{Fricke elements in the Monster VOA}\label{s:Fricke}

In this section, let $V\equiv V^\natural$ be the FLM Monster holomorphic VOA with $c=24$  \cite{FLM1988}, i.e. the only known (and, conjecturally, the unique) holomorphic VOA with $c=24$ and no currents, $V^\natural_1=0$. Its automorphism group is isomorphic to the Monster group $\Aut(V^\natural)\cong \mathbb{M}$, the largest sporadic finite simple group.  Let $g\in \Aut(V^\natural)$ be a Fricke non-anomalous element of order $N$ \cite{ConwayNorton1979}. This means that the corresponding McKay-Thompson series
\be T_g(\tau)\equiv T_{g^0,g^1}(\tau):=\Tr_{V^\natural}(gq^{L_0-1})
\ee is invariant under the  Fricke involution
\be T_g(-\frac{1}{N\tau})=T_g(\tau)\ .
\ee This implies that the partition function of the $g$-twisted sector is
\be T_{g,1}(\tau)=q^{-1/N}+0+O(q^{1/N})\ .
\ee In particular, there is a single $g$-twisted ground state with conformal weight $1-\frac{1}{N}$. Thus, theorem \ref{th:paraf} implies that for every such Fricke element $g$ of order $N$, there is a parafermion algebra $\paraf(N)$ embedded in the $g$-fixed subVOA $V^g$, and $V$ is self-orbifold $V\cong V/\langle g\rangle$. There are $80$ such conjugacy classes, up to algebraic conjugation, namely
\begin{multline}
1A,\ 2A,\ 3A,\ 4A,\ 5A,\ 6A,\ 6B,\ 6C,\ 6D,\ 7A,\ 8A,\ 9A,\ 10A,\ 10B,\ 10C,\ 10D,\ 11A,\\ 
12A,\ 12B,\ 12E,\ 12H,\ 13A,\ 14A,\ 14B,\ 14C,\ 15A,\ 15B,\ 15C,\ 16C,\ 17A,\ 18A,\ 18B,\ 18C,\ 18E,\\ 19A,\ 20A,\ 20C,\ 20F,\ 21A,\ 21B,\ 21D,\ 22A,\ 22B,\ 23AB,\ 24B,\ 24C,\ 24I,\ 25A,\ 26A,\ 26B,\ 27AB,\\ 28B,\ 28C,\ 29A,\ 30A,\ 30B,\ 30C,\ 30D,\ 30F,\ 30G,\ 31AB,\ 32A,\ 33A,\ 33B,\ 34A,\ 35A,\ 35B,\\ 36A,\ 36B,\ 36D,\ 38A,\ 39A,\ 39CD,\ 41A,\ 42A,\ 42B,\ 42D,\ 44AB,\ 45A,\ 46AB,\ 46CD,\ 47AB,\\ 50A,\ 51A,\ 52A,\ 54A,\ 55A,\ 56A,\ 59AB,\ 60B,\ 60C,\ 60D,\ 62AB,\ 66A,\ 66B,\ 69AB,\\ 70A,\ 70B,\ 71AB,\ 78A,\ 78BC,\ 87AB,\ 92AB,\ 94AB,\ 95AB,\ 105A,\ 110A,\ 119AB\ .
\end{multline}
For each such (non-trivial)class $X$ with $X\in\{2A,3A,\ldots\}$, we denote by
\be \calC_X\equiv \calC_X^\natural:=\Com(\paraf(N),V^g)
\ee the commutant in $V^g$ of the parafermion algebra $\paraf(N)$ associated with an element $g$ in class $X$.
The $g$-fixed subVOA $V^g\subset V$ carries a representation of the group $G_g:=C_{\MM}(g)/\langle g\rangle$, where 
\be C_{\MM}(g)=\{h\in \MM\mid hg=gh\}
\ee is the centralizer of $g$ in the Monster group $\MM$. The group $G_g$ acts projectively on the $V^g$-modules $V_{a,b}$; more precisely, the group acting linearly is the central extension $\langle g,Q\rangle.G_g$ of $G_g$ by $\langle g,Q\rangle\cong \ZZ_N\times \ZZ_N$ where $g$ and the quantum symmetry $Q$ act on $V_{a,b}$ by multiplication by $e^{2\pi i\frac{b}{N}}$ and $e^{2\pi i\frac{a}{N}}$, respectively. Our construction implies that the group $G_g$ acts trivially on the parafermion algebra $\paraf(N)\subset V^g$. It follows that the group of automorphisms of the commutant $\calC_X$ contains $G_g$. 

The presence of a parafermion algebra was already known for some Monster conjugacy classes of small order $N$, such as $N=2$, $N=3$, and $N=5$ \cite{HohnLamYamauchi,HoehnLamYamauchi2012McKayE6,LamYamadaYamauchi}. For class $2A$ and $3A$, the commutant subalgebra have been studied in good detail; in particular, $\calC_{2A}$ is the Baby Monster subVOA of $V^\natural$ \cite{HohnLamYamauchi}, while $\calC_{3A}$ is a VOA with automorphism group (containing) the Fischer group $F_{3+}=Fi'_{24}$ \cite{HoehnLamYamauchi2012McKayE6}.

In general, the $g$-fixed subVOA will contain the parafermion algebras related to all pairs $g^{\pm k}$ of elements of $\langle g\rangle$ that are Fricke, in particular for all $\pm k$ that are coprime to $N$. For example, for $g$ in class $71AB$, the $g$-fixed subalgebra contains $35$ different parafermion algebras $\paraf(71)$ of central charge $c=\frac{2N-2}{N+2}=\frac{140}{73}$, corresponding to all possible pairs $g^{\pm k}$ for $k=1,\ldots, 35$. These algebras $\paraf(N)$ in general do not commute with each other, with some exception that will be described in sections \ref{s:5A} and \ref{s:pairparaf}. It would be interesting to study the VOA generated by such algebras, and their commutant. 

\subsection{An example: class 5A} \label{s:5A}

Let us focus on the $5A$ case, that has been recently considered in \cite{Bae:2020pvv,Honda:2026bjy}. The McKay-Thompson series for $g$ in this class is \cite{ConwayNorton1979}
\be T_g(\tau)=\frac{\eta(\tau)^6}{\eta(5\tau)^6}+6+125\frac{\eta(5\tau)^6}{\eta(\tau)^6}=\frac{1}{q}+134 q+760 q^2+3345 q^3+12256 q^4+39350 q^5+114096 q^6
+\ldots
\ee Using the procedure described in the previous section, we obtain the $q$-expansions of the characters $f_{[l,t]}$ of the modules $\calC_{5A}([l,t])$ for the algebra $\calC_{5A}=\Com(\paraf(5),V^g)$:
\begin{align*}
f_{[0,0]}&=f_{[5,5]}=q^{-\frac{20}{21}}\left(1+27228 q^2+2447296 q^3+86616936 q^4+1834812656 q^5
+\ldots\right)\\
f_{[2,0]}&=f_{[3,5]}=q^{\frac{16}{21}}\left(12122+1725712 q+74865624 q^2+1793410208 q^3+29502447295 q^4
+\ldots\right)\\
f_{[4,0]}&=f_{[1,5]}=q^{\frac{4}{21}}\left(133+101840 q+7660876 q^2+247745408 q^3+4961433786 q^4
+\ldots\right)\\
f_{[1,\pm 1]}&=f_{[4,\mp 4]}=q^{\frac{104}{105}}\left(36005+3519256 q+129371722 q^2+2805498048 q^3+43106152834 q^4
q^5
+\ldots\right)\\
f_{[3,\pm 1]}&=f_{[2,\mp 4]}=q^{\frac{59}{105}}\left(3344+727738 q+37321832 q^2+983838676 q^3+17258872128 q^4
+\ldots\right)\\
f_{[5,\pm 1]}&=f_{[0,\mp 4]}=q^{-\frac{16}{105}}\left(1+8912 q+1096738 q^2+44579664 q^3+1027651312 q^4+16467862536 q^5
+\ldots\right)\\
f_{[0,\pm 2]}&=f_{[5,\mp 3]}=q^{\frac{26}{105}}\left(134+74880 q+5277993 q^2+164952568 q^3+3232862752 q^4
+\ldots\right)\\
f_{[2,\pm 2]}&=f_{[3,\mp 3]}=q^{\frac{101}{105}}\left(38456+3915900 q+146786912 q^2+3221277802 q^3+49904924808 q^4
+\ldots\right)\\
f_{[4,\pm 2]}&=f_{[1,\mp 3]}=q^{\frac{41}{105}}\left(760+267710 q+16150304 q^2+465675085 q^3+8664020368 q^4
+\ldots\right)
\end{align*}
The centralizer $C_\MM(g)$ of an element $g$ in class $5A$ is isomorphic to $\ZZ_5\times HN$, where $HN$ is the Harada-Norton simple group \cite{Atlas}. Therefore $G_g=C_\MM(g)/\langle g\rangle\cong HN$ acts by automorphisms on $\calC_{5A}$.

As discussed in \cite{LamYamadaYamauchi} and \cite{Bae:2020pvv}, the VOA $V^g$ for $g$ in class 5A actually contains two commuting copies of the $\paraf(5)$ parafermion algebra. Indeed, one can see from the character above that the $\calC_{5A}$-modules $\calC_{5A}([0,\pm 4])$ contain a single ground state of weight $1-1/5$. Furthermore, the module $\calC_{5A}([0,\pm 4])$ is a simple current of order $5$ for the VOA $\calC_{5A}$. Therefore, we can apply theorem \ref{th:simplecurrent} to conclude that the VOA $\calC_{5A}$ contains a second $\ZZ_5$ parafermion algebra, that we denote as $\paraf'_5$. In fact, it is easy to see that this second $\ZZ_5$ parafermion is associated with the ground states of weight $1-1/5$ in the $V^g$-modules $V_{2,2}$ and $V_{-2,-2}$. The decomposition of the $V^g$-modules characters into characters of $\paraf(5)\times \calC_{5A}$ in this case gives
\be \tilde T_{\pm 2,\pm 2}(\tau)=\chi^{\paraf(5)}_{[0,0]}(\tau)f_{[0,\pm 4]}(\tau)+\chi^{\paraf(5)}_{[2,0]}(\tau)f_{[2,\pm 4]}(\tau)+\chi^{\paraf(5)}_{[4,0]}(\tau)f_{[4,\pm 4]}(\tau)
\ee and one can observe that the ground state of $V_{2,2}$ is in the component $\paraf(5,[0,0])\otimes \calC_{5A}([0,\pm 4])$. Because it is in the vacuum representation of $\paraf(5)$, the new parafermion algebra $\paraf(5)'$ built out of it commutes with $\paraf(5)$.

One can therefore define a smaller subVOA in $\calC_{5A}$ given by the commutant of $\paraf(5)'\subset \calC_{5A}$, i.e
\be \calC_{5,5}:=\Com(\paraf(5)',\calC_{5A})=\Com(\paraf(5)\otimes \paraf(5)',V^g)\ .
\ee
One can compute the characters of this new subVOA $\calC_{5,5}$ using exactly the same method that we used to determine the $\calC_{5A}$-characters $f_{[l,t]}$. With respect to the order $5$ permutation given by the simple current $\calC_{5A}([0,+4])$, the $15$ irreducible $\calC_{5A}$-modules organize into $3$ orbits, that can be labeled by $l\in \{0,2,4\}$. For each such orbit, there is a module of the extended $\tilde V$ algebra in theorem \ref{th:paraf}, which character can be written in different ways as
\be \sum_{t\in \ZZ/5\ZZ} \frac{\Theta^{(5)}_{2t}(\tau,z)}{\eta(\tau)} f_{[l,4t]}(\tau)=\sum_{l'\in\{0,2,4\}}\ch^{\alg{su}(2)_5}_{l'}(\tau,z)F_{l',l}(\tau)=\sum_{l'\in\{0,2,4\}}\sum_m \frac{\Theta^{(5)}_m(\tau,z)}{\eta(\tau)} \chi^{\paraf(5)}_{[l',m]}(\tau)F_{l',l}(\tau)\ ,
\ee from which one can obtain the characters $F_{l',l}\equiv F_{5-l',5-l}$, $l,l'\in \{0,2,4\}$ of $\calC_{5,5}$. We also get the decomposition
\be f_{[l,4t]}(\tau)=\sum_{l'\in\{0,2,4\}}\chi^{\paraf(5)}_{[l',2t]}(\tau)F_{l',l}(\tau)\ ,\qquad l\in\{0,2,4\},\ t\in \ZZ/5\ZZ\ ,
\ee and
\be \tilde T_{a,b}=\sum_{\substack{l,l'=0\\ l,l'\equiv a+b\bmod 2}}^5\chi^{\paraf(5)}_{[l,a-b]}(\tau)\chi^{\paraf(5)}_{[l',-2(a+b)]}(\tau)F_{l',l}(\tau)\ .
\ee The partition function of $V^\natural$ decomposes as
\be Z_{V^\natural}(\tau)=\sum_{b\in \ZZ/5\ZZ} \tilde T_{0,b}(\tau)=\sum_{b\in \ZZ/5\ZZ} \sum_{\substack{l,l'=0\\ l,l'\equiv b\bmod 2}}^5\chi^{\paraf(5)}_{[l,-b]}(\tau)\chi^{\paraf(5)}_{[l',-2b]}(\tau)F_{l',l}(\tau)\ .
\ee
We report the first terms in the $q$-expansions of the $\calC_{5,5}$-characters:
\begin{align*}
F_{0,0}(\tau)&=F_{5,5}(\tau)=q^{-\frac{19}{21}}(1+18316 q^2+1360096 q^3+42393826 q^4+811613728 q^5
+\ldots)\\
F_{2,0}(\tau)&=F_{3,5}(\tau)=q^{\frac{17}{21}}(8778+1003408 q+37866696 q^2+815035704 q^3+12259464259 q^4
+\ldots)\\
F_{4,0}(\tau)&=F_{1,5}(\tau)=q^{\frac{5}{21}}(133+65968 q+4172476 q^2+119360584 q^3+2166248140 q^4
+\ldots)\\
F_{0,2}(\tau)&=F_{5,3}(\tau)=q^{\frac{17}{21}}(8778+1003408 q+37866696 q^2+815035704 q^3+12259464259 q^4
+\ldots)\\
F_{2,2}(\tau)&=F_{3,3}(\tau)=q^{\frac{11}{21}}(3344+680504 q+32364068 q^2+795272512 q^3+13076464336 q^4
+\ldots)\\
F_{4,2}(\tau)&=F_{1,3}(\tau)=q^{\frac{20}{21}}(35112+3184818 q+108781232 q^2+2204347320 q^3+31813496792 q^4
+\ldots)\\
F_{0,4}(\tau)&=F_{5,1}(\tau)=q^{\frac{5}{21}}(133+65968 q+4172476 q^2+119360584 q^3+2166248140 q^4
+\ldots)\\
F_{2,4}(\tau)&=F_{3,1}(\tau)=q^{\frac{20}{21}}(35112+3184818 q+108781232 q^2+2204347320 q^3+31813496792 q^4
+\ldots)\\
F_{4,4}(\tau)&=F_{1,1}(\tau)=q^{\frac{20}{21}}(760+231705 q+12595936 q^2+333082540 q^3+5746222592 q^4
+\ldots)
\end{align*}
The $q$-expansions of these functions coincide with the characters of the VOA $VHN^\natural$ given in section 3.2.5 of \cite{Bae:2020pvv}, as expected since $VHN^\natural$ can be obtained as the commutator of the product $\paraf(5)\otimes \paraf(5)'$ of the two parafermion algebras. This is a independent consistency check of our calculations, given that the functions in \cite{Bae:2020pvv}were obtained using different methods that make use of the modularity properties.

Notice also that the normalizer
\be N_\MM(g):=\{h\in\MM\mid hgh^{-1}=g^k\text{ for some }k\}
\ee contains an involution that exchanges $g$ and $g^2$. Such an involution must exchange the two parafermion algebras $\paraf(5)$ and $\paraf(5)'$, so it must map $F_{l,l'}$ to $F_{l',l}$. Thus, it acts on $\calC_{5,5}$ by an outer automorphism.

\subsection{Pairs of commuting parafermions in the Monster VOA} \label{s:pairparaf}

The construction of two commuting parafermions $\paraf(N)\otimes \paraf(N)'$ we just described for the class $5A$ can be generalized to several conjugacy classes of Fricke non-anomalous symmetries, namely
\be
5A,\ 10D,\ 13A,\ 17A,\ 26B,\ 29A,\  41A\ .
\ee
For these numbers $N$, there exist $r\in \ZZ/N\ZZ$ such that $r^2=-1\mod N$, so that the ground state of the $V^g$-module $V_{r,r}$ has conformal weight $1-1/N$. Furthermore, these classes have the property that $1-1/N$ is the lowest possible conformal weight for a $V^g$-module $V_{a,b}$ different from the vacuum. Let us specialize the three different formulae \eqref{tildeVchars}, \eqref{su2exp}, and\eqref{tildeVchars2}  for the character of the $\tilde V$-module $\tilde V_t$ to the case $t=2r$:
\be \sum_{n\in \ZZ/N\ZZ} \frac{\Theta^{(N)}_{2n-2r}(\tau,z)}{\eta(\tau)}\tilde T_{n,2r-n}(\tau)=\sum_{\substack{l=0\\ l\text{ even}}}^N\ch^{\alg{su}(2)_N}_l(\tau,z)f_{[l,2r]}(\tau)=\sum_{\substack{l=0\\l \text{ even}}}^N\sum_{\substack{m\in \ZZ/2N\ZZ\\ m\text{ even}}}\frac{\Theta^{(N)}_{m}(\tau,z)}{\eta(\tau)}\chi^{\paraf(N)}_{[l,m]}(\tau)f_{[l,2r]}(\tau)\ .
\ee From the first expression, it is clear that the lowest power in the $q$-expansion of this character is a single state with zero $\alg{u}(1)$-charge from the $\frac{\Theta^{(N)}_{0}(\tau,z)}{\eta(\tau)}\tilde T_{r,r}(\tau)$ term in the sum -- all the other states have strictly larger conformal weight.  In the second expression, this term must come from the $l=0$ contribution $\ch^{\alg{su}(2)_N}_0(\tau,z)f_{[0,2r]}(\tau)$, since all the other $\alg{su}(2)_N$ characters have degenerate ground states with non-zero charge. As a consequence, this $q$-power must come from the $\chi^{\paraf(N)}_{[0,0]}(\tau)f_{[0,2r]}(\tau)$ term in the third sum, and in particular from the vacuum of the parafermion algebra $\paraf(N)$. 
 It follows that the $\calC_X$-module $\calC_X([0,2r])$ has a ground state of conformal weight $1-1/N$. Furthermore, $\calC_X([0,2r])$ is a simple current of order $N$, because
\be \calC_X([0,2r])\boxtimes \calC_X([l,t])\cong \calC_X([l,t+2r])\ ,\qquad l\in \{0,\ldots,N\},\ t\in \ZZ/2N\ZZ
\ee and the property $r^2=-1\mod N$ implies that $2r$ has order $N$ mod $2N\ZZ$.
Therefore, one can apply theorem \ref{th:paraf} to conclude that $\calC_X$ contains a second $\ZZ_N$ parafermion algebra $\paraf(N)'$. We denote by
\be \calC_{N,N}:=\Com(\paraf(N)',\calC_X)=\Com(\paraf(N)\otimes \paraf(N)',V^g)
\ee the commutant of $\paraf(N)'$ in $\calC_X$. By the same argument as the $\ZZ_5$ case above, the characters $F_{l',l}$ of the $\calC_{N,N}$-modules can be labeled by a pair of indices $l,l'\in \{0,\ldots,N\}$ with $l\equiv l'\mod 2$, and with the field identification $F_{l,l'}=F_{N-l,N-l'}$. The functions $F_{l,l'}$ can be determined in terms of $f_{[l,t]}$ using
\be \sum_{t\in \ZZ/N\ZZ} \frac{\Theta^{(N)}_{2t}(\tau,z)}{\eta(\tau)} f_{[l,2rt]}(\tau)=\sum_{\substack{l'=0\\l' \text{ even}}}^N\ch^{\alg{su}(2)_N}_{l'}(\tau,z)F_{l',l}(\tau)=\sum_{\substack{l'=0\\l' \text{ even}}}^N\sum_m \frac{\Theta^{(N)}_m(\tau,z)}{\eta(\tau)} \chi^{\paraf(N)'}_{[l',m]}(\tau)F_{l',l}(\tau)\ ,
\ee from which one can also obtain  the decompositions
\be f_{[l,2rt]}(\tau)=\sum_{\substack{l'=0\\l' \text{ even}}}^N\chi^{\paraf(N)'}_{[l',2t]}(\tau)F_{l',l}(\tau)\ ,\qquad l\in\{0,2,\ldots, N-1\},\ t\in \ZZ/N\ZZ\ .
\ee Some small manipulations using $r^2=-1\mod N$ and the  field identifications yield
\be f_{[l,s]}(\tau)=\sum_{\substack{l'=0\\l'\equiv s\bmod{2} }}^N\chi^{\paraf(N)'}_{[l',-rs]}(\tau)F_{l',l}(\tau)\ ,\qquad l\in\{0,1,\ldots, N\},\ s\in \ZZ/2N\ZZ,\ s\equiv l\bmod 2\ ,
\ee so that
\be \tilde T_{a,b}=\sum_{\substack{l,l'=0\\ l,l'\equiv a+b\bmod 2}}^N\chi^{\paraf(N)}_{[l,a-b]}(\tau)\chi^{\paraf(N)'}_{[l',-r(a+b)]}(\tau)F_{l',l}(\tau)\ .
\ee The partition function of $V^\natural$ decomposes as
\be Z_{V^\natural}(\tau)=\sum_{b\in \ZZ/N\ZZ} \tilde T_{0,b}(\tau)=\sum_{b\in \ZZ/N\ZZ} \sum_{\substack{l,l'=0\\ l,l'\equiv b\bmod 2}}^N\chi^{\paraf(N)}_{[l,-b]}(\tau)\chi^{\paraf(N)'}_{[l',-rb]}(\tau)F_{l',l}(\tau)\ .
\ee The centralizer $C_\MM(g)/\langle g\rangle$ of $g$ in the Monster group leaves the parafermion algebras $\paraf(N)$ and $\paraf(N)'$ fixed, so it acts on $\calC_{N,N}$ by automorphisms. More generally, for each $g$ in one of the classes above, there is an element $h$ in the normalizer $N_\MM(g)$ such that $hgh^{-1}=g^r$. This element exchanges the two parafermionic algebras $\paraf(N)$ and $\paraf(N)'$, and act on $\calC_{N,N}$ as an outer automorphism exchanging the $(l,l')$ and the $(l',l)$ module. One consequence of this automorphism is the character identity $F_{l',l}=F_{l,l'}$.

\subsection{Parafermions in the Leech lattice VOA and in Schellekens theories}\label{s:paraLeech}

Let us consider the Leech lattice VOA $V_\Lambda$; it is a holomorphic VOA of $c=24$ with $24$ currents generating $\alg{u}(1)^{24}$.  Given an automorphism $\nu\in \Aut(\Lambda)\cong Co_0$ of the Leech lattice $\Lambda$, we denote by
\be \Lambda^\nu:=\{\lambda\in \Lambda\mid \nu(\lambda)=\lambda\}
\ee the invariant sublattice, and by
\be \Lambda_\nu:=\{\mu\in \Lambda\mid \mu\cdot \lambda=0\ \forall \lambda\in \Lambda^\nu\}=(\Lambda^\nu)^\perp\cap \Lambda
\ee its orthogonal complement (the coinvariant sublattice).
Every such lattice automorphism $\nu\in \Aut(\Lambda)$ lifts to a VOA automorphism $\hat \nu\in \Aut(V_\Lambda)$ of the same order $N$.\footnote{This is quite a special property of the VOA $V_\Lambda$ and its automorphism group. In a general lattice VOA $V_L$, it might happen that for a given lattice automorphism $g\in \Aut(L)$ of order $N$, the standard lifts $\hat g\in \Aut(V_L)$ to VOA automorphisms have order at least $2N$ (order doubling).} Furthermore, one can choose the lift $\hat \nu$ of order $N$ such that it acts trivially on the subVOA $V_{\Lambda^\nu}$ associated with the $\nu$-fixed sublattice $\Lambda^\nu$. 
More precisely, given the conformal embedding \be V_{\Lambda^\nu}\otimes V_{\Lambda_\nu} \subset V_\Lambda\ee of the product of the two commuting subVOAs $V_{\Lambda^\nu}$ and $V_{\Lambda_\nu}$, one has that $\hat \nu$ restricts to the identity on $V_{\Lambda^\nu}$, while it acts on $V_{\Lambda_\nu}$ by an automorphism that preserves no current, $(V_{\Lambda_\nu})^{\hat\nu}_1=0$. We also have a conformal embedding of $\hat\nu$-fixed VOAs
\be\label{hatnuembed} V_{\Lambda^\nu}\otimes (V_{\Lambda_\nu})^{\hat\nu} \subset (V_\Lambda)^{\hat\nu}\ .
\ee 

The characteristic polynomial of a Conway element $\nu$ of order $N$ in the $24$-dimensional representation of $Co_0$ has the form
\be \prod_{\ell|N} (\nu^\ell -1)^{m_\ell}\ ,
\ee where the product is over all divisors $\ell$ of $N$, and  $m_\ell\in \ZZ$ are (possibly negative) integers with $\sum_{\ell|N} \ell m_\ell=24$.  Such numbers are usually symbolically represented by the `Frame shape' $\prod_{\ell|N} \ell^{m_\ell}$ of $\nu$; it is known that the Frame shape univocally determines the algebraic conjugacy class of $\nu$ in $Co_0$. For a $\nu$ with Frame shape $\prod_{\ell|N} \ell^{m_\ell}$, the $\hat\nu$-twisted and $\hat\nu$-twining partition functions are
\be\label{twistwinLeech} Z^{\hat\nu}(\tau,\xi)=\frac{\Theta_{\Lambda^\nu}(\tau,\xi)}{\prod_{\ell |N}\eta(\ell\tau)^{m_\ell}}\qquad Z_{\hat\nu}(\tau,\xi)=\sqrt{\frac{\prod_{\ell|N}\ell^{m_\ell}}{|(\Lambda^\nu)^*/\Lambda^\nu|}}\frac{\Theta_{(\Lambda^\nu)^*}(\tau,\xi)}{\prod_{\ell |N}\eta(\tau/\ell)^{m_\ell}}
\ee where the fugacities $\xi\in (\Lambda^\nu\otimes\RR)/\Lambda^\nu$ in the lattice theta series
\be \Theta_{L}(\tau,\xi)=\sum_{\lambda\in L}q^{\frac{\lambda^2}{2}}e^{2\pi i \xi\cdot\lambda}\ ,
\ee keep track of the charges with respect to the $\hat\nu$-invariant $\alg{u}(1)$ currents in $V_{\Lambda^\nu}$. From the expansion $\Theta_{(\Lambda^\nu)^*}(\tau,\xi)=1+\ldots$, it follows that the $\hat\nu$-twisted ground states are always in the trivial representation with respect to the $\hat \nu$-preserved current algebra $(V_\Lambda)^{\hat\nu}_1=(V_{\Lambda^\nu})^{\hat\nu}_1$. Thus, one of the conditions of theorem \ref{th:paraf} is satisfied.

For all conjugacy classes in $Co_0$ for which $\hat\nu$ is not anomalous, we computed the conformal weight of the $\hat\nu$-twisted sector. Remarkably, the only possible values that one obtains for a symmetry of order $N$ are $1-\frac{1}{N}$, $1$, $1+\frac{1}{N}$. The $40$ conjugacy classes with conformal weight  $1+\frac{1}{N}$ are exactly the ones for which the orbifold $V/\langle \hat\nu\rangle$ is the Monster VOA $V^\natural$. We are interested in the $52$ classes ($50$ algebraic classes) for which the conformal weight is $1-\frac{1}{N}$, that are
\begin{align}\label{paraclasses}
&1^{24},\ 1^8 2^8,\ \frac{1^{12} 6^{12}}{2^{12} 3^{12}},\ 1^6 3^6,\ \frac{1^8
	4^8}{2^8},\ 1^4 2^2 4^4,\ \frac{1^6 10^6}{2^6 5^6},\ 1^4 5^4,\ \frac{1^4 2^1 
	6^5}{3^4},\ \frac{1^5 3^1 6^4}{2^4},\ 1^2 2^2 3^2 6^2,\ \frac{1^4 14^4}{2^4 7^4},\ 1^3
7^3,\nonumber\\& \frac{1^4 8^4}{2^2 4^2},\ 1^2 2^1 4^1 8^2,\ \frac{1^3 18^3}{2^3 9^3},\ \frac{1^3
	9^3}{3^2},\ \frac{1^2 2^1  10^3}{5^2},\ \frac{1^3 5^1  10^2}{2^2},\ 1^2 11^2,\ \frac{1^4
	12^4}{3^4 4^4},\ \frac{1^2  3^2 4^2 12^2}{2^2 6^2},\ \frac{1^1 2^2
	3^1 12^2}{4^2},\ \frac{1^2 4^1 6^2 12^1}{3^2},\nonumber\\ &\frac{1^3 12^3}{2 3 4 6},\ \frac{1^2
	26^2}{ 2^213^2},\ 1^1 2^1 7^1 14^1,\ \frac{1^3 15^3}{3^3 5^3},\ \frac{1^2 5^2
	6^2 30^2}{ 2^2 3^210^2 15^2},\ 1^1 3^1 5^1 15^1,\ \frac{1^2 15^2}{3^1 5^1},\ \frac{1^2 16^2}{2^1
	8^1},\ \frac{1^2 9^1 18^1}{2^1 3^1},\ \frac{1^1 2^1 18^2}{6^1 9^1},\nonumber\\ &\frac{1^2 20^2}{4^2 5^2},\ \frac{1^1 2^1
	10^1  20^1}{4^1 5^1},\ \frac{1^2 21^2}{3^2 7^2},\ \frac{1^1 3^1  14^1 42^1}{2^1 6^1 7^1 21^1},\ 1^1 23^1,\ 1^1
23^1,\ \frac{1^2 4^1 6^1 24^2}{ 2^1 3^2 8^2 12^1},\ \frac{1^1 4^1 6^1 24^1}{3^1 8^1},\ \frac{1^1 4^1 7^1 28^1}{
	2^1 14^1},\\ &\frac{1^1 2^1  15^1 30^1}{ 3^1 5^1 6^1 10^1},\ \frac{1^1 6^1 10^1 15^1}{3^1 5^1},\ \frac{2^1 3^1 5^1 30^1}{
	6^1 10^1},\ \frac{1^1 6^1 11^1 66^1}{2^1 3^1 22^1 33^1},\ \frac{1^1 35^1}{5^1 7^1},\ \frac{1^1 36^1}{4^1 9^1},\ \frac{1^1
	39^1}{ 3^1 13^1},\ \frac{1^1 39^1}{ 3^1 13^1},\ \frac{1^1 4^1 6^1 10^1 15^1 60^1}{2^1 3^1 5^1 12^1 20^1 30^1}\nonumber
\end{align} In these cases, theorem \ref{th:paraf} applies, so that $V_\Lambda$ is self-orbifold under the corresponding $\hat\nu$, and  there is an embedding of the $\ZZ_N$ parafermion algebra $\paraf(N)$ in $(V_\Lambda)^{\hat\nu}$. In fact, one can verify by a case by case calculation, using the data about the fixed point sublattices of $\Lambda$ in \cite{HoehnMason2016}, that for all these Frame shapes the ratio $\frac{\prod_{\ell|N}\ell^{m_\ell}}{|(\Lambda^\nu)^*/\Lambda^\nu|}$ appearing in \eqref{twistwinLeech} is $1$, so that there is a unique $\hat\nu$-twisted ground state with conformal weight $1-1/N$. This implies that the extended VOA $\tilde V$ defined in the proof of theorem \ref{th:paraf} contains exactly the three currents generating $\alg{su}(2)_N$ plus the ones in $V^{\hat\nu}$, and none more. It is tempting to conjecture that this non-degeneracy is a general feature of theories where theorem \ref{th:paraf} applies, because it seems difficult to obtain a consistent current algebra in $\tilde V$ otherwise.

We can actually be more precise about the embedding of $\paraf(N)$ in $(V_\Lambda)^{\hat\nu}$. Because the $\hat\nu$-twisted ground states are in the vacuum representation with respect to the subVOA $V_{\Lambda^\nu}$ appearing in \eqref{hatnuembed}, it follows that the algebra $\paraf(N)$ embeds in $(V_{\Lambda_\nu})^{\hat\nu}$. If we denote by
\be \calC_X^{\Lambda}:= \Com(\paraf(N),(V_\Lambda)^{\hat\nu})\qquad \calC_X^{\Lambda_\nu}:= \Com(\paraf(N),(V_{\Lambda_\nu})^{\hat\nu})
\ee the commutants of $\paraf(N)$ in $V_\Lambda$ and $V_{\Lambda_\nu}$, where $X$ is one of the Frame shapes above, then
there are conformal embeddings
\be  V_{\Lambda^\nu}\otimes \paraf(N)\otimes \calC_X^{\Lambda_\nu} \quad \subset\quad V_{\Lambda^\nu}\otimes(V_{\Lambda_\nu})^{\hat\nu} \quad \subset \quad (V_{\Lambda})^{\hat\nu}\quad \subset \quad V_\Lambda\ .
\ee 
For later reference, let us describe in more detail the decomposition of $V$ and the sectors $V_{a,b}$ into $(V_{\Lambda_\nu})^{\hat\nu}\otimes V_{\Lambda^\nu}$-modules. For simplicity, let us focus on the case where $N$ is prime, so that the $\nu^a$ invariant and coinvariant lattices are always the same, $\Lambda^{\nu^a}=\Lambda^\nu$ and $\Lambda_{\nu^a}=\Lambda_\nu$, for all $a\neq 0\mod N$.  From the general properties of even unimodular lattices and their automorphisms, we know that the Leech lattice decomposes as
\be\label{gluelattice} \Lambda=\bigcup_{\gamma\in (\Lambda^\nu)^*/\Lambda^\nu} (\gamma+ \Lambda^\nu)\oplus (\iota(\gamma)+\Lambda_\nu)\ ,
\ee where $\iota:(\Lambda^\nu)^*/\Lambda^\nu \to (\Lambda_\nu)^*/\Lambda_\nu$ is an isomorphism of abelian groups, see for example \cite{Nikulin1980}. There is a
corresponding decomposition of $V_\Lambda$ into modules for the sublattice VOAs $V_{\Lambda^\nu}\otimes V_{\Lambda_\nu}$ \be\label{latticeVOAdecomp} V_\Lambda=\bigoplus_{\gamma\in (\Lambda^\nu)^*/\Lambda^\nu} V_{\gamma+ \Lambda^\nu}\otimes V_{\iota(\gamma)+\Lambda_\nu}\ ,
\ee where $V_{\gamma+ \Lambda^\nu}$ and $V_{\iota(\gamma)+\Lambda_\nu}$ have conformal weights in $\frac{\gamma^2}{2}+\ZZ$ and $-\frac{\gamma^2}{2}+\ZZ$, respectively.

Because $\nu\in \Aut(\Lambda)$ fixes $\Lambda^\nu$ pointwise, the induced action on $(\Lambda^\nu)^*/\Lambda^\nu$ must be trivial as well. Compatibility with \eqref{latticeVOAdecomp} implies that the restriction $\hat \nu_{\rvert V_{\Lambda_\nu}}$ maps each $V_{\Lambda_\nu}$-module $V_{\iota(\gamma)+\Lambda_\nu}$ to itself. This implies that there is a decomposition similar to \eqref{latticeVOAdecomp} for all $\hat \nu^a$-twisted modules 
\be V_\Lambda(\hat\nu^a)=\bigoplus_{\gamma\in (\Lambda^\nu)^*/\Lambda^\nu} V_{\gamma+ \Lambda^\nu}\otimes V_{\iota(\gamma)+\Lambda_\nu}(\hat\nu^a)\ ,
\ee and for all $V_\Lambda^{\hat\nu}$-modules $V_{a,b}\subset V_\Lambda(\hat\nu^a)$ 
\be\label{Lambdanudecomp} V_{a,b} = \bigoplus_{\gamma\in (\Lambda^\nu)^*/\Lambda^\nu} V_{\gamma+ \Lambda^\nu}\otimes V_{\iota(\gamma)+\Lambda_\nu,a,b}\ .
\ee
Here, $V_{\gamma+ \Lambda^\nu}$ are the same (untwisted)  $V^{\Lambda^\nu}$-modules as in \eqref{latticeVOAdecomp}, $V_{\iota(\gamma)+\Lambda_\nu}(\hat\nu^a)$  denote the $\hat \nu^a$-twisted $V^{\Lambda_\nu}$-modules, and $V_{\iota(\gamma)+\Lambda_\nu,a,b}$ are the eigenspaces of $\hat\nu$ in each $V_{\iota(\gamma)+\Lambda_\nu}(\hat\nu^a)$.
Notice that, for $a\neq 0\mod N$, the action of $\hat\nu$ is defined by
\be V_{\iota(\gamma)+\Lambda_\nu,a,b}:=\CC\left\{v\in V_{\iota(\gamma)+\Lambda_\nu}(\hat\nu^a)\mid h(v)\in \frac{ab}{N}-\frac{\gamma^2}{2}+\ZZ\right\}\ ,\qquad a\neq 0\mod N\ ,
\ee  so as to match our conventions that the conformal weights of $V_{a,b}$ are in $\frac{ab}{N}+\ZZ$.

Finally, the decomposition of each $V_{a,b}$ into modules of $V_{\Lambda^\nu}\otimes \paraf(N)\otimes \calC_X^{\Lambda_\nu} $  reads
\be V_{a,b} = \bigoplus_{\gamma\in (\Lambda^\nu)^*/\Lambda^\nu} \bigoplus_{\substack{l=0\\l\equiv a+b\bmod{2}}}^N V_{\gamma+ \Lambda^\nu}\otimes \paraf(N,[l,a-b])\otimes  \calC_X(\iota(\gamma)+\Lambda_\nu,[l,a+b]) \ ,
\ee  so that
\be\label{VLeechparaf} V_\Lambda=\bigoplus_{b\in \ZZ/N\ZZ} V_{0,b}(\tau)=\bigoplus_{b\in \ZZ/N\ZZ} \bigoplus_{\gamma\in (\Lambda^\nu)^*/\Lambda^\nu}\bigoplus_{\substack{l=0\\l\equiv b\bmod{2}}}^NV_{\gamma+ \Lambda^\nu}\otimes \paraf(N,[l,-b])\otimes  \calC_X(\iota(\gamma)+\Lambda_\nu,[l,b]) \ .
\ee Here, $\calC_X(\iota(\gamma)+\Lambda_\nu,[l,a+b]) $ are the $ \calC_X^{\Lambda_\nu}$-modules such that
\be\label{finedecomp} V_{\iota(\gamma)+\Lambda_\nu,a,b}= \bigoplus_{\substack{l=0\\l\equiv a+b\bmod{2}}}^N \paraf(N,[l,a-b])\otimes  \calC_X(\iota(\gamma)+\Lambda_\nu,[l,a+b])\ . \ee
The characters $f_{\iota(\gamma)+\Lambda_\nu,[l,a+b]}(\tau)$ for the modules $\calC_X(\iota(\gamma)+\Lambda_\nu,[l,a+b])$ can be easily obtained by repeating the procedure in section \ref{s:commutant}, while keeping track of the charges under  $\hat\nu$-invariant $\alg{u}(1)$ currents in $V_{\Lambda^\nu}$. In particular,
\be \tilde T_{a,b}(\tau,\xi )=\sum_{\substack{l=0\\l\equiv a+b\bmod{2}}}^N \chi^{\paraf(N)}_{[l,a-b]}(\tau)f_{[l,a+b]}(\tau,\xi )\  .
\ee where 
\be f_{[l,a+b]}(\tau,\xi )=\sum_{\gamma\in (\Lambda^\nu)^*/\Lambda^\nu}\frac{\Theta_{\gamma+\Lambda^\nu}(\tau,\xi)}{\eta(\tau)^{\dim \Lambda^\nu}}f_{\iota(\gamma)+\Lambda_\nu,[l,a+b]}(\tau)\ ,
\ee 
are the characters of the $ \calC_X^{\Lambda}$-modules
\be \calC_X([l,a+b])=\bigoplus_{\gamma\in (\Lambda^\nu)^*/\Lambda^\nu}V_{\gamma+ \Lambda^\nu}\otimes\calC_X(\iota(\gamma)+\Lambda_\nu,[l,a+b])\ .
\ee

This has a number of interesting consequences for other holomorphic VOAs. First of all, it was proved in \cite{MollerScheithauer2023,MollerScheithauer2024} that, besides the Monster module $V^\natural$, all the other $70$ holomorphic VOAs of central charge $24$ with currents (classified by Schellekens in \cite{Schellekens1993})  can be obtained from $V_\Lambda$ by orbifolding a cyclic group $\langle \hat\nu h\rangle$ corresponding to a `generalized deep hole'. Here, the generator $\hat\nu h$ is the product of the standard lift $\hat \nu$ of a lattice automorphism $\nu\in \Aut(\Lambda)$, and an inner automorphism  $h\in (\Lambda\otimes \RR)/\Lambda\cong U(1)^{24}$ acting non-trivially only on $V_{\Lambda^\nu}$, and that can be identified with an element of order $N$ in $(\Lambda^\nu\otimes \QQ)/\Lambda^\nu$. In particular, there are $11$ $Co_0$-conjugacy classes of automorphisms $\nu$ involved in this construction, namely\footnote{These $11$ algebraic classes also play a role in two other  uniform constructions of Schellekens theories, see \cite{Hoehn2017GenusMoonshine,HohnMoller2022}.}
\be\label{goodnus} 1^{24},\quad 1^82^8,\quad 1^63^6,\quad 1^42^24^4,\quad 1^45^4,\quad 1^22^23^26^2,\quad 1^37^3,\quad 1^22^14^18^2,
\ee
\be 2^{12},\quad 2^36^3,\quad 2^210^2\ .
\ee 
 For three of these classes, $2^{12}$, $2^36^3$, and $2^210^2$, the automorphism $\hat\nu$, by itself, is anomalous; the product $\hat\nu h$ is not anomalous, but its order is twice the order of $\hat\nu$. Excluding these three cases and the identity, for the remaining $7$ classes the conformal weight of the $\hat\nu$ twisted sector is $1-1/N$. The symmetry $\hat\nu h$ restricts to an automorphism of the subVOA $V_{\Lambda^\nu}\otimes V_{\Lambda_\nu}$, with $\hat\nu$ acting non-trivially only on $V_{\Lambda_\nu}$ and $h$ acting non-trivially only on $V_{\Lambda^\nu}$. Therefore, the orbifold theory $V_\Lambda/\langle \hat\nu h\rangle$ contains $(V_{\Lambda^\nu})^h\otimes (V_{\Lambda_\nu})^{\hat\nu}$ as a conformally embedded subVOA. But we have already proved that the $\ZZ_N$ parafermion algebra can be embedded in $(V_{\Lambda_\nu})^{\hat\nu}$. We conclude that the orbifold theory contains $\paraf(N)$ as a subVOA, commuting with the subVOA generated by the currents.
 
 More into detail, let $V_S$ denote Schellekens' theory obtained from a generalized-deep-hole orbifold $V_S\cong V_\Lambda/\langle \hat\nu h\rangle$, where $\nu$ one of the $7$ non-trivial classes in \eqref{goodnus}. Then, one gets the decomposition
 \be\label{Schellorb} V_S=\bigoplus_{a,b\in \ZZ/N\ZZ} \bigoplus_{\gamma\in (\Lambda^\nu)^*/\Lambda^\nu} V_{\gamma+ \Lambda^\nu,a,-b}\otimes V_{\iota(\gamma)+\Lambda_\nu,a,b}\ .
 \ee Here, $V_{\iota(\gamma)+\Lambda_\nu,a,b}$ are the same $(V_{\Lambda_\nu})^{\hat\nu}$-modules appearing in \eqref{Lambdanudecomp}, and $V_{\gamma+ \Lambda^\nu,a,b}$ are the $(V_{\Lambda^\nu})^{h}$-modules obtained as the $h=e^{\frac{2\pi i b}{N}}$ eigenspace in the $h^a$-twisted sector. For consistency with our previous conventions, the action of $h$ on the $h^a$-twisted sectors must be defined in such a way that the conformal weights of $V_{\gamma+ \Lambda^\nu,a,b}$ take values in $\frac{ab}{N}+\frac{\gamma^2}{2}+\ZZ$. From \eqref{Schellorb} and \eqref{finedecomp}, one immediately gets the decomposition of Schellekens theory $V_S$ into modules of $\paraf(N)$ and its commutant $\Com(\paraf(N),V_S)$ in $V_S$
\be\label{VSdecomp}  V_S=\bigoplus_{a,b\in \ZZ/N\ZZ} \bigoplus_{\gamma\in (\Lambda^\nu)^*/\Lambda^\nu} \bigoplus_{\substack{l=0\\l\equiv a+b\bmod{2}}}^N  V_{\gamma+ \Lambda^\nu,a,-b}\otimes \paraf(N,[l,a-b])\otimes  \calC_X(\iota(\gamma)+\Lambda_\nu,[l,a+b])\ .\ee It follows that\footnote{One might be worried that the parafermionic vacuum might also appear in the $l=N$ term in \eqref{VSdecomp}, due to the field identification $\paraf(N,[N,\pm N])\cong \paraf(N,[0,0])$. However, if we take $a,b$ in the same range $\{0,\ldots, N-1\}$, their difference $a-b$ is always different from $\pm N$. }
\be \Com(\paraf(N),V_S)\cong  \bigoplus_{a=0}^{N-1} \bigoplus_{\gamma\in (\Lambda^\nu)^*/\Lambda^\nu}  V_{\gamma+ \Lambda^\nu,a,-a}\otimes   \calC_X(\iota(\gamma)+\Lambda_\nu,[0,2a]) \ .
\ee
 
\subsection{An example at central charge $32$}

Let us now describe an example of  holomorphic VOA $V$ of central charge $c=32$ that is self-orbifold with respect to cyclic groups of order $3$ or $7$. The theory $V$ is obtained as the lattice VOA for a $32$-dimensional even unimodular lattice $L$, whose roots (elements of square length $2$) generate a sublattice $E_6\subset L$ isomorphic to the $E_6$ root lattice.  The orthogonal complement $T_{26}:=L\cap E_6^\perp$ of the root lattice  is the unique (up to isomorphism) $26$-dimensional even lattice with determinant $3$ and with shortest non-zero vectors of square norm $4$, and was studied in \cite{Borcherds1984Leech}. Therefore, the currents of $V$ generate an affine Kac-Moody algebra $\mathfrak{e}_{6,1}\oplus \alg{u}(1)^{26}$. As in eq.\eqref{gluelattice}, the unimodular lattice $L$ is obtained as the union of $3$ cosets
\be L=T_{26}\oplus E_6\ \cup\ (v+T_{26}\oplus E_6)\ \cup\ (-v+T_{26}\oplus E_6)\ ,
\ee where $v\in T_{26}^*\oplus E_6^*$ is a suitable `glue vector' whose square norm is an even integer and such that $3v$ is the smallest multiple of $v$ contained in $T_{26}\oplus E_6$. This means that $V$ decomposes into the sum of three irreducible modules for the subVOA $V_{T_{26}}\otimes V_{E_6}$, where $V_{E_6}$ is generated by the $\alg{e}_{6,1}$ currents.

The automorphism group of $L$ is $( {}^3D_4(2).3\times W(E_6)).\ZZ_2$, where the $\ZZ_3$-extended twisted Chevalley group ${}^3D_4(2).3$ acts trivially on the $E_6$ root lattice, the Weyl group $W(E_6)$ acts trivially on $T_{26}$, and the $\ZZ_2$ is the usual reflection in all $32$ directions.

The group ${}^3D_4(2).3$ contains elements $\nu_3$ and $\nu_7$ of order $3$ and $7$, respectively, whose characteristic polynomials on the $26$-dimensional representation are
\be (\nu_3-1)^8(\nu_3^3-1)^6\ ,\qquad (\nu_7-1)^5(\nu_7^7-1)^3\ .
\ee Being of odd order, their standard lifts $\hat\nu_3$ and $\hat\nu_7$ to symmetries of the VOA $V$ have the same order as the lattice automorphisms. The conformal weights of the $\nu_3$-twisted and $\nu_7$-ground states can be computed from their eigenvalues on the $26$ free bosons of the $T_{26}$ lattice component, and are $\frac{2}{3}$ and $\frac{6}{7}$, respectively. Such ground states are necessarily in the trivial representation of the $\alg{e}_{6,1}$ Lie algebra, because $\hat\nu_3$ and $\hat\nu_7$ act trivially on the subVOA $V_{E_6}$. They are also neutral with respect to the $\hat\nu_3$ or $\hat\nu_7$ invariant $\alg{u}(1)$ currents -- this is always true for standard lifts of lattice automorphisms.

Therefore, theorem \ref{th:paraf} tells us that $V$ is self-orbifold with respect to both $\hat \nu_3$ and $\hat \nu_7$, and that there are embeddings of the $\ZZ_3$ and $\ZZ_7$ parafermion algebras $\paraf(3)$ and $\paraf(7)$ in the $\hat \nu_3$- and $\hat\nu_7$-fixed subalgebras of $V$. Furthermore, the VOA $V$ admits fusion categories of topological defects associated with the $\paraf(3)$ and $\paraf(7)$ parafermion algebras, and whose preserved subVOAs are, respectively, the dual pairs of commuting subVOAs $\paraf(3)\times \mathcal{C}_3$ and $\paraf(7)\times \mathcal{C}_7$, where
\be \mathcal{C}_3=\Com(\paraf(3),V)\ ,\qquad \mathcal{C}_7=\Com(\paraf(7),V)\ .
\ee The commutants $\mathcal{C}_3$ and $\mathcal{C}_7$ of the parafermions algebras are VOAs with central charges
\be c(\mathcal{C}_3)=33-\frac{9}{5}=31+\frac{1}{5}\ ,\qquad c(\mathcal{C}_7)=33-\frac{7}{3}=30+\frac{2}{3}\ .
\ee The currents of  $\mathcal{C}_3$ and $\mathcal{C}_7$ generate the affine algebras $\mathfrak{u}(1)^{14}\oplus \alg{e}_{6,1}$ and $\mathfrak{u}(1)^{8}\oplus \alg{e}_{6,1}$, respectively.

\subsection{A non-holomorphic example from heterotic strings}

Let us consider an example of non-holomorphic conformal field theory with central charges $(c,\tilde c)=(20,6)$, arising from the compactification of heterotic string on $T^4$. We take the convention where the supersymmetric sector is anti-holomorphic. The CFT contains $20$ holomorphic free scalars and $4$ anti-chiral free scalars and free fermions. The Narain lattice of winding-momenta is an indefinite even unimodular lattice $\Gamma$ of signature $(4,20)$. The moduli space of such compactifications is
\be (O(4)\times O(20))\backslash O(4,20)/O(4,20,\ZZ)\ .
\ee This space can be interpreted as a Grassmannian parametrising $4$-dimensional positive definite real subspaces $\Pi_R$ in the real space $\Gamma\otimes \RR\cong \RR^{4,20}$, quotiented by the discrete T-duality group $O(4,20,\ZZ)=\Aut(\Gamma)$, the automorphism group of the lattice $\Gamma$. Given such a $\Pi_R$, let $\Pi_L=\Pi_R^\perp$ denote its orthogonal complement of signature $(0,20)$ in $\RR^{4,20}$. The choice of $\Pi_R\subset \Gamma\otimes\RR$ determines the decomposition of each winding-momentum vector $p\in \Gamma$ into its holomorphic (left-moving) and antiholomorphic (right-moving) components $p=(p_L;p_R)$, with $p_L\in \Pi_L$ and $p_R\in \Pi_R$. The partition function is given by
\be Z_s(\tau,\bar\tau)=\frac{\Theta_{\Gamma,\Pi_R}(\tau,\bar\tau)}{\eta(\tau)^{20}\eta(\bar\tau)^4}\frac{\theta_s(\bar\tau)^2}{\eta(\bar\tau)^2}
\ee where $s$ labels the $4$ different spin structures on the torus, and the moduli dependence is encoded in the Narain theta series
\be \Theta_{\Gamma,\Pi_R}(\tau,\bar\tau)=\sum_{p=(p_L;p_R)\in \Gamma}q^{\frac{p_L^2}{2}}\bar q^{\frac{p_R^2}{2}}\ .
\ee At a generic point in the moduli space, the holomorphic current algebra is $\alg{u}(1)^{20}$. It gets enhanced to a non-abelian current algebra at special points where $\Pi_R$ is orthogonal to a vector $p\in \Gamma$ with square norm $-2$ (a root), where the corresponding vertex operators are holomorphic with conformal weight $1$.

When two subspaces $\Pi_R$ and $\Pi_R'$ are related by a transformation in $O(4,20,\ZZ)$, then the corresponding CFTs are related by T-duality, and therefore they are equivalent. Given a certain $\Pi_R\equiv\Pi $, the setwise stabilizer $\tilde G_{\Pi}=\mathrm{Stab}_\Pi\subset O(4,20,\ZZ)$ of the subspace $\Pi$ corresponds to the group of self-dualities, so it gives rise to (invertible) symmetries of the corresponding CFT. The subgroup $G_\Pi\subset \tilde G_\Pi$ that fixes $\Pi$ pointwise, rather than setwise, corresponds to symmetries that fix the anti-holomorphic (large) $\CN=4$ superconformal algebra. If we restrict ourselves to points on the moduli space where the holomorphic current algebra is the generic $\alg{u}(1)^{20}$, then the classification of the possible groups $G_\Pi$ is essentially the same as the classification of the groups of symmetries of non-linear sigma models on K3 in \cite{Gaberdiel:2011fg}. 

The main result of \cite{Gaberdiel:2011fg} is that, for all $\Pi$, the group $G_\Pi$ is always finite, and it is isomorphic  to a subgroup $H$ of the Conway group $Co_0=\Aut(\Lambda)$ that fixes a sublattice of $\Lambda$ of rank at least $4$ (a classification of the subgroups of $Co_0$ arising in this way is given in \cite{HoehnMason2016}). Vice versa, every such lattice stabilizer $H\subset \Aut(\Lambda)$ arises as a symmetry group of a heterotic model at some point in the moduli space with generic current algebra. This correspondence works not just at the level of abstract groups: the actions of $G_\Pi \subset \Aut(\Gamma)$ and $H\subset \Aut(\Lambda)$ on the respective $24$-dimensional lattices are closely related. In particular, if we denote by $\Gamma^G$ and $\Lambda^H$ the (pointwise) fixed sublattices, and  by $\Gamma_G=\Gamma\cap (\Gamma^G)^\perp$ and $\Lambda_H=\Lambda\cap (\Lambda^H)^\perp$ the respective orthogonal complements, one has an isomorphism
\be \Gamma_G\cong \Lambda_H(-1)\ ,
\ee of coinvariant lattices that reverses the quadratic form, and that relates the action of $G$ on $\Gamma_G $ and the action of $H$ on $\Lambda_H$ \cite{Gaberdiel:2011fg}. In particular, each $g\in G$ has the same eigenvalues in the $24$-dimensional representation $\Gamma\otimes \RR$ as the corresponding $\nu \in H$.

Let us choose an element $\nu\in \Aut(\Lambda)$ whose Frame shape $X=\prod_{\ell|N}  \ell^{m_\ell}$ is one of the $52$ eq.\eqref{paraclasses}, and such that the dimension $\sum_{\ell|N} m_\ell$ of the $\nu$-fixed sublattice is greater than $4$. By the previous remarks, there is a point in the moduli space of heterotic compatifications on $T^4$ with a symmetry $g\in G_\Pi\subset \Aut(\Gamma)$, such that $g$ has the same Frame shape as $\nu$, and the coinvariant lattices are the essentially the same, $\Gamma_g\cong \Lambda_\nu(-1)$. The fact that $g$ is a symmetry for this model means that the $4$-subspace $\Pi\subset \Gamma\otimes \RR$ is pointwise fixed by $g$. Equivalently, this means that $\Pi\subseteq \Gamma^g\otimes \RR$, and this is true if and only if $\Gamma_g\in \Pi^\perp$. This means that for all vectors $p\in \Gamma_g$, the splitting with respect to $\Pi_R\equiv\Pi$ and $\Pi_L\equiv \Pi^\perp$ is of the form $p=(p_L;0)$. Thus, all vertex operators $\mathcal{V}_p$, $p\in \Gamma_g$, are purely holomorphic, and generate a subVOA of the chiral algebra of this model. But because $\Gamma_g\cong \Lambda_\nu(-1)$, this subVOA is exactly the lattice VOA $V_{\Lambda_\nu}$ appearing in section \ref{s:paraLeech}. One can then just repeat the same construction as in that section, to find that the chiral algebra of every heterotic model with such a symmetry $g$ of order $N$ contains a parafermion algebra $\paraf(N)$. The NS or Ramond Hilbert space $\CH_{NS/R}$ of the heterotic model decomposes as
\be \CH_{NS/R}=\bigoplus_{b\in \ZZ/N\ZZ} \bigoplus_{\gamma\in (\Gamma^g)^*/\Gamma^g}\bigoplus_{\substack{l=0\\l\equiv b\bmod{2}}}^N\CH_{NS/R,\gamma+ \Gamma^g}\otimes \paraf(N,[l,-b])\otimes  \calC_X(\iota(\gamma)+\Lambda_\nu,[l,b]) \ .
\ee 
where $\calC_X(\iota(\gamma)+\Lambda_\nu,[l,b]) $ are exactly the same modules as in \eqref{VLeechparaf}, while the holomorphic $V_{\Lambda^\nu}$-modules $V_{\gamma+ \Gamma^g}$ in \eqref{VLeechparaf} are replaced by some `non-holomorphic' $\CH_{NS/R,\gamma+ \Gamma^g}$. More precisely, $\CH_{NS/R,\gamma+ \Gamma^g}$ includes the (NS or R) representations with winding-momentum $p\in \gamma+ \Gamma^g$, for the algebra generated by the $4$ antiholomorphic free fermions, and the chiral and antichiral free bosons corresponding to the (possibly indefinite) lattice $\Gamma^g$.

As in section \ref{s:topdefparaf}, these heterotic models admit a fusion category of topological defects that is isomorphic to the modular tensor category of $\paraf(N)$. The simple defects $\CL_{[l,t]}$ have twining partition function
\be Z^{[l,t]}_s(\tau,\bar\tau)=\sum_{b\in \ZZ/N\ZZ} \sum_{\substack{\ell=0\\\ell\equiv b\bmod{2}}}^N \lambda_{[l,t],[\ell ,b]}\chi^{\paraf(N)}_{[\ell ,b]}(\tau)f_{s,[\ell,-b]}(\tau,\bar\tau)\ .
\ee Here, $s$ is the spin structure, and
\be f_{s,[l,t]}(\tau,\xi )=\sum_{\gamma\in (\Gamma^g)^*/\Gamma^g}\frac{\Theta_{\gamma+\Gamma^g}(\tau,\bar\tau)}{\eta(\tau)^{\dim \Gamma^g-4}\eta(\bar\tau)^{4}}\frac{\theta_s(\bar\tau)^2}{\eta(\bar\tau)^2}f_{\iota(\gamma)+\Lambda_\nu,[l,t]}(\tau)\ , 
\ee with $f_{\iota(\gamma)+\Lambda_\nu,[l,t]}(\tau)$ the same characters as in section \ref{s:paraLeech}.

To the best of our knowledge,  many of the topological defects of heterotic models obtained in this way were not known before. A large class of topological defects in toroidal (super)string compactifications were described in \cite{Bachas:2012bj}. However, all defects in \cite{Bachas:2012bj} preserve all the (holomorphic and anti-holomorphic) $\alg{u}(1)$ currents of the model, possibly up to an automorphism, and we checked by direct calculation that this is not true for most of the $\CL_{[l,t]}$ with $l\neq 0$.

\section{Twisted sector in a non-trivial representation of the current algebra}\label{s:extensions}

The theorems stated in the previous section assume that the simple current module contains a vector that is in the trivial representation of the current algebra generated by the weight $1$ fields in the VOA. Even in the case where this is not true, several interesting results can be derived, provided that the conformal weight of the twisted sector is smaller than $1$. We illustrate this generalization in some examples.

\subsection{Self-orbifold symmetries in the $E_8$ lattice VOA}

In this section, let $V\equiv V_{E_8}$ denote the bosonic $E_8$ lattice VOA. The cyclic groups with respect to which $V$ is self-orbifold were recently studied in \cite{Burbano:2021loy} using Lie-algebraic methods. A general strategy to determined the corresponding duality defects was also provided, and several examples were described explicitly.

Here, we show how a variant of our method can be used to obtain similar results.

\subsubsection{The non-anomalous  $\ZZ_2$ symmetry}

Let us consider an automorphism of order $2$ of $V$. It is known that there is only one conjugacy class of automorphisms $g\in \Aut(V_{E_8})\cong E_8$ that have order $2$ and are non-anomalous. The $g$-fixed subVOA is the lattice VOA $V_{D_8}$, generated by the currents of $\alg{so}(16)_1$ current algebra
\be V^g\equiv V_{0,0}=\alg{so}(16)_1,\ .
\ee Recall that $\alg{so}(16)_1$ has four irreducible representations that we label as $L_{\alg{so}(16)}(1,x)$ where $x$ is either $0$ (vacuum), $v$ (vector), $s$ (spinor) or $c$ (conjugate spinor). The eigenspace of $g$ with eigenvalue $-1$ is the $s$ representation $V_{0,1}\cong L_{\alg{so}(16)}(1,s)$, while the $g$-twisted sector $V_{1,0}\oplus V_{1,1}$ has components
\be V_{1,0}\cong  L_{\alg{so}(16)}(1,c)\ ,\qquad   V_{1,1}\cong L_{\alg{so}(16)}(1,v)\ .
\ee In particular, there are $8$ $g$-twisted ground states of conformal weight $1/2$ in $V_{1,1}$. While none of the twisted ground states are in the vacuum representation of the current algebra $\alg{so}(16)_1$, we can still apply our general procedure with $N=2$. We introduce a $c=1$ lattice VOA $W_0=V_{2\ZZ}$ based on the lattice $\sqrt{2N}\ZZ=2\ZZ$, with modules $W_m=V_{\frac{m}{2}+2\ZZ}$ labeled by $m\in \ZZ/4\ZZ$ of conformal weight $\frac{m^2}{8}$. We extend the product $V_{0,0,0}=W_0\times V_{0,0}$ by the weight $1$ simple current $V_{2,1,1}=W_2\times V_{1,1}$
\be \tilde V=V_{0,0,0}\oplus V_{2,1,1}\ .
\ee The modules $\tilde V_l$, $l\in\ZZ/4\ZZ$ are given by
\be \tilde V_0\cong \tilde V\qquad \tilde V_1=V_{1,0,1}\oplus V_{-1,1,0}\qquad \tilde V_{-1}=V_{-1,0,1}\oplus V_{1,1,0}\qquad \tilde V_2=V_{2,0,0}\oplus V_{0,1,1}\ .
\ee The $V_{0,0,0}=W_0\times V_{0,0}$ component of $\tilde V$ contains the $\alg{u}(1)\cong \alg{so}(2)$ current algebra of $W_0$ and the $\alg{so}(16)_1$ of $V_{0,0}$. The $V_{2,1,1}$ contains $32$ additional currents transforming in the $L_{\alg{so}(2)}(1,v)\otimes L_{\alg{so}(16)}(1,v)$ representation of $\alg{so}(2)_1\, \alg{so}(16)_1$. Altogether, the current algebra of $\tilde V$ is, therefore, $\alg{so}(18)_1$, and the four modules $\tilde V_0$, $\tilde V_1$, $\tilde  V_{-1}$, $\tilde V_2$ are the $0$, $s$, $c$, $v$ representations of $\alg{so}(18)_1$. The algebra $\tilde V$ has a group $SO(18)/\ZZ_2$ of inner automorphisms. Consider any involution $h\in SO(18)$ that flips the sign of the $\alg{so}(2)$ current in $W_0$ and maps the $\alg{so}(16)_1$ subalgebra to itself. This involution must be contained in the $S(O(2)\times O(16))$ subgroup of $SO(18)$, and both projections $h_2\in O(2)$ and $h_{16}\in O(16)$ must have negative determinant. Being an inner automorphism of $\tilde V$, $h$ must map each $\tilde V_l$ to itself, and it must exchange the components $V_{m,a,b}$ with opposite $\alg{so}(2)$ charge $m$. In particular $h$ maps $V_{\pm 1, 0,1}=W_{\pm 1}\otimes V_{0,1}$ to $V_{\mp 1,1,0}=W_{\mp 1}\otimes V_{1,0}$. Thus $h_{16}$ is an automorphism of $V_{0,0}\cong V^g\cong \alg{so}(16)_1$ that exchanges $V_{0,1}$ and $V_{1,0}$, and therefore provides an isomorphism between $V=V_{0,0}\oplus V_{0,1}$ and $V/\langle g\rangle=V_{0,0}\oplus V_{1,0}$.

For each choice of $h\equiv (h_2,h_{16})$, there is a corresponding duality defect operator $\hat \CN$ that annihilates the $V_{0,1}$ component of $V$ and acts by $\sqrt{2} h_{16}$ on $V_{0,0}$. The twining partition function is
\be \Tr_{V}(\hat \CN q^{L_0-\frac{c}{24}})=\sqrt{2}\Tr_{V_{0,0}}(h_{16}q^{L_0-\frac{c}{24}})\ .
\ee While it is not too difficult to compute this twining partition function directly, one can also use an argument analogous to \eqref{dualitytwin} to get the formula
\be \Tr_{V}(\hat \CN q^{L_0-\frac{c}{24}})=\sqrt{2}\frac{\eta(2\tau)}{\eta(\tau)}\ch^{\alg{so}(18)_1}_0(\tau,\xi^{(r)})\ ,
\ee
where the $\alg{so}(18)_1$ vacuum character is given by
\be \ch^{\alg{so}(18)_1}_0(\tau,\xi)=\frac{1}{\eta(\tau)^9}\sum_{\substack{n_1,\ldots,n_9\in \ZZ\\ n_1+\ldots + n_9\in 2\ZZ}}q^{\frac{1}{2}(n_1^2+\ldots n_9^2)}e^{2\pi i (n_1\xi_1+\ldots+n_9\xi_9)}\ ,\qquad \xi=(\xi_1,\ldots,\xi_9)\in \RR^9\ ,
\ee and
\be \xi^{(r)}=(\underbrace{\frac{1}{2}\,\ldots,\frac{1}{2}}_{r\text{ times}},0,\ldots,0)\ ,\qquad 1\le r\le 9\ .
\ee Here, $2r$ the number of $-1$ eigenvalues of the $SO(18)$ inner involution $h$. By varying the value of $r$, one gets only $4$ different duality defects, because $\xi^{(r)}$ or $\xi^{(9-r)}$ give the same result.


\subsubsection{$\ZZ_3$ symmetry and all the rest}

Let us try to apply a similar argument to the non-anomalous $\ZZ_3$ symmetry of $V_{E_8}$. The fixed subVOA is $V^g=V_{0,0}=V_{A_2\oplus E_6}$. The $g$- and $g^{-1}$-twisted sectors have $3$ states of conformal weight $1/3$, respectively in the $(3,1)$ and $(\bar 3,1)$ representation of $\alg{su}(3)_1\alg{e}_{6,1}$, and $27$ states of conformal weight $2/3$ respectively in the $(1,27)$ and $(1,\bar{27})$ representation of  $\alg{su}(3)_1\alg{e}_{6,1}$. Thus, once again we have twisted states with conformal weight $1-1/N$ with $N=3$.

By applying the usual construction with $N=3$, we consider $W_0$ to be the VOA based on the lattice $\sqrt{6}\ZZ$, and extend $V_{0,0,0}:=W_0\otimes V_{0,0}$ by the simple current $V_{2,1,-1}=W_1\otimes V_{1,-1}$ of weight $1$, to get
\be \tilde V=V_{0,0,0}\oplus V_{2,1,-1}\oplus V_{-2,-1,1}\ .
\ee The VOA $\tilde V$ contains the $8+78$ currents of $\alg{su}(3)_1\alg{e}_{6,1}$ in $V_{0,0}$, $1$ current in $W_0$, and $27+27$ currents in  $V_{2,1,-1}\oplus V_{-2,-1,1}$ in the $(1,27)\oplus (1,\bar{27})$ representation of $\alg{su}(3)_1\alg{e}_{6,1}$. Altogether, the $8+133$ currents generate a $\alg{su}(3)_1\alg{e}_{7,1}$ algebra. The $\tilde V$-modules are given by
\be \tilde V_l=V_{-l,0,l}\oplus V_{-l+2,1,l-1}\oplus V_{-l-2,-1,l+1}\ .
\ee
 The $E_7$ group of inner automorphisms contains an involution that flips the sign of the current in $W_0$ and restricts to an automorphism of $\alg{e}_{6,1}$. For example, one can take the lift to $E_7$ of the longest element in the Weyl group, that acts by minus the identity on the Cartan algebra, and has order $2$ in its adjoint action on $\alg{e}_{7,1}$ currents (the square is the generator of the $\ZZ_2$ centre of $E_7$). This involution exchanges the $V_{0,n}$ sectors with $V_{n,0}$, so that the $E_8$ VOA is self-orbifold with respect to $g$, and the $E_7$ involution provides the corresponding duality defect.

A variant in this construction is to obtain new currents from the twisted sector states of conformal weight $1/3$, i.e. $1-r/N$ with $r=2$ and $N=3$. To this aim, we take $W_0$ to be the lattice VOA for $\sqrt{2rN}\ZZ=\sqrt{12}\ZZ$, define as usual $V_{0,0,0}:=W_0\otimes V_{0,0}$, and extend by the simple current $V_{2r,1,-r}=W_{4}\otimes V_{1,1}$ of conformal weight $1$ to get
\be \tilde V=V_{0,0,0}\oplus V_{4,1,1}\oplus V_{-4,-1,-1}\ .
\ee The components $V_{2,1,-1}\oplus V_{-2,-1,1}$ contain $3+3$ currents in the $(3,1)\oplus (\bar 3,1)$ of $\alg{su}(3)_1\alg{e}_{6,1}$, that extend the current algebra of $V_{0,0,0}$ to $\alg{su}(4)_1\alg{e}_{6,1}$. The $\tilde V$-modules are given by 
\be \tilde V_l=\bigoplus_{n\in \ZZ/N\ZZ} V_{2rn-l,n,l-2rn}=V_{-l,0,l}\oplus V_{-l+4,1,l+1}\oplus V_{-l-4,-1,l-1}\ ,\qquad l\in \ZZ/12\ZZ\ .
\ee Notice that, contrary to the cases considered in theorem \ref{th:paraf} and in the previous examples of this section, the set of $W_0$ charges contained in the modules $\tilde V_l$ is not symmetric with respect the origin. There is no inconsistency here: contrary to the groups $SU(2)$, $SO(8)$, and $E_7$ considered in the previous cases, the $SU(4)$ group of inner automorphisms of $\tilde V$ does \emph{not} contain any involution that changes the signs of its Cartan generators -- this is simply the statement that $SU(4)$ admits representations that are not self-conjugate. This is the reason why the argument in the proof of theorem \ref{th:paraf} fails in this case, and the conformal weight $h(v)$ is not equal to $1-1/N$. Of course, there is still an \emph{outer} charge conjugation automorphism of $SU(4)$ that exchanges the modules in a suitable way, but exactly because it is outer, it is not particularly useful to understand the self-orbifold properties of the model, or get an nice formula for the duality defect.

Not everything is lost, though. We have just realized that the embedding $\alg{su}(3)_1$ in $\alg{su}(4)_1$ is not really useful, but this is not the only possibility: going one step further in the chain of embeddings \be \alg{su}(3)_1\subset \alg{su}(4)_1\subset \alg{so}(7)_1\ ,\ee we find the $SO(7)$ group that has the same rank $3$ as $SU(4)$, but contains a suitable Cartan involution. The coset $\frac{\alg{so}(7)_1}{\alg{su}(4)_1}=\paraf(2)=Vir_{c=1/2}$ is just the Virasoro algebra at $c=1/2$, with irreducible modules $M_h$, $h\in \{0,\frac{1}{2},\frac{1}{16}\}$. This suggests how to proceed: consider the tensor product VOA $V'_{0,0,0}=Vir_{c=1/2}\otimes W_0\otimes V_{0,0}$, where again $W_0=V_{\sqrt{12}\ZZ}$, and extend it by the order $6$ simple current $V'_{2,1,1}=M_{\frac{1}{2}}\otimes W_2\otimes V_{-1,-1}$ of conformal weight $\frac{1}{2}+\frac{4}{2\cdot 12}+\frac{1}{3}=1$, to  get
\be \tilde V'=\bigoplus_{n\in \ZZ/6\ZZ} V'_{2n,n,n}=(M_{0}\otimes\tilde V) \oplus (M_{\frac{1}{2}}\otimes \tilde V_{6})\ .
\ee This shows that $\tilde V'$ contains $\tilde V$ as a subVOA, and therefore the $\alg{su}(4)_1$ affine algebra. However, it also contains $6$ more currents in the $V'_{2,1,1}$ and $V'_{-2,-1,-1}$ components, extending $\alg{su}(4)_1$ to $\alg{so}(7)_1$. The standard simple current construction implies that the modules of $\tilde V'$ are either of the form
\be \tilde V'_{t}=\bigoplus_{n\in \ZZ/6\ZZ} M_{\frac{n\bmod 2}{2}}\otimes W_{t+2n}\otimes V_{-n,-t-n}\ ,\qquad t\in \ZZ/12\ZZ,\ t\text{ even}
\ee or
\be
\tilde V'_{t}=\bigoplus_{n\in \ZZ/6\ZZ} M_{\frac{1}{16}}\otimes W_{t+2n}\otimes V_{-n,-t-n}\ ,\qquad t\in \ZZ/6\ZZ,\ t\text{ odd}\ .
\ee 
 It is now easy to check, using the same methods as in the proof of theorem \ref{th:paraf}, that the Cartan involution of $SO(7)$ maps each $\tilde V'_t$ module to itself, while flipping the sign of the $\alg{u}(1)$-current in $W_0$, and therefore it exchanges the $n$ and the $-t-n$ components in $\tilde V'_t$. This means that the induced action on $V_{0,0}$ is an outer automorphism exchanging each module $V_{a,b}$ with $V_{-b,-a}$, thus proving that the theory is self-orbifold.

Of course, in this particular example, it is probably easier to derive the same results using the more standard methods described in \cite{Burbano:2021loy} -- this is because the $g$-fixed VOA $V^g$ is very well understood, and we know already that there is an outer automorphism of $V^g$ exchanging the twisted sectors in the way expected for a self-orbifold theory.  Nevertheless, this construction illustrates how some variants of the techniques of section \ref{s:mainth} can be applied in much more general contexts.

\subsection{Far-from-the-hole orbifolds from the Leech VOA to itself}

Theorem \ref{th:paraf} applies when the $g$-twisted ground state $v$ is neutral with respect to the $g$-invariant currents in $V$ and has conformal weight $1-1/N$. In the proof of the theorem, one takes the tensor product with a state $\CV_{\sqrt{\frac{2}{N}}}$ with conformal weight $1/N$ and $\alg{u}(1)$ charge $\sqrt{\frac{2}{N}}$, so as to get a $\alg{su}(2)_N$ current. Vice versa, as we will now explain, when  the $g$-twisted ground state $v$ is charged with respect to some $\alg{u}(1)$ current, and has conformal weight $1/N$, one can get a $\alg{su}(2)_N$ by applying the construction in opposite direction, i.e. by taking the tensor product with a $\paraf(N)$ parafermion. This construction yields the following theorem.

\begin{theorem}\label{th:invparaf}
	Let $V$ be a strongly rational holomorphic VOA, $g\in \Aut(V)$ be a non-anomalous automorphism of finite order $N>2$. Let $\alg{h}=\bigoplus_i \alg{g}_{i}\oplus \alg{u}(1)^r$ be the current algebra generated by the weight $1$ fields in the $g$-fixed subVOA $V^g$, where $\alg{g_i}$ are (non-abelian) simple affine algebras. If the $g$-twisted sector $V_g$ contains a non-zero vector $v\in V_g$ such that:
	\begin{enumerate}
		\item the conformal weight is $h(v)=\frac{1}{N}$,
		 and
		\item $v$ has charge $Q(v)=\sqrt{\frac{2}{N}}$ with respect to a $\alg{u}(1)$ current $j(z)$ in the abelian component $\alg{u}(1)^r$ of $\alg{h}$ and is in the vacuum representation with respect to the non-abelian components $\alg{g_i}$,
	\end{enumerate} then, $V$ is self-orbifold with respect to $g$, i.e. the orbifold VOA $V/\langle g\rangle$ is isomorphic to $V$. 
\end{theorem}
Here, $Q(v)$ is the $j_0$-eigenvalue of $v$, where $j(z)$ has the standard normalization $j(z)j(0)\sim 1/z^2$ for $\alg{u}(1)$ currents.
\begin{proof}
	Consider the tensor product of $V^g$ with the $\ZZ_N$ parafermion algebra
	\be \paraf(N)\otimes V^g \ .
	\ee Its module $\paraf(N,[0,2])\otimes V_{1,1}$ is a simple current of order $N$ for this VOA with conformal weight $1-\frac{1}{N}+\frac{1}{N}=1$. Then, we can consider the simple current extension
	\be \tilde V=\bigoplus_{n\in \ZZ/N\ZZ} \paraf(N,[0,2n])\otimes V_{n,n}
	 \ ,
	\ee with modules
	\be \tilde V_{\ell,t}=\bigoplus_{n\in \ZZ/N\ZZ} \paraf(N,[\ell,t+2n])\otimes V_{t+n,n}=\bigoplus_{\substack{m\in \ZZ/2N\ZZ\\ m\equiv t\mod 2}} \paraf(N,[\ell,m])\otimes V_{\frac{m+t}{2},\frac{m-t}{2}}
	\ .
	\ee The VOA $\tilde V$ contains two weight $1$ fields $j^\pm(z)$ in the $\paraf(N,[0,\pm 2])\otimes V_{\pm 1,\pm 1}$ components, that have charge $\pm\sqrt{\frac{2}{N}}$ with respect to $j(z)$ and commute with all the other currents in $V^g$. Therefore, $j^+(z)$, $j^-(z)$ and $j(z)$ generate a $\alg{su}(2)_N$ affine algebra, and $\tilde V$ admits a $SU(2)$ group of inner automorphisms. Let $T_{\alg{u}(1)}(z)=\frac{1}{2}\normord{j(z)j(z)}$, $T_V(z)$,  $T_{\tilde V}(z)$ and $T_{\paraf(N)}(z)$ denote the stress-tensors of the respective VOAs. Then, $T_\calC(z):=T_V(z)-T_{\alg{u}(1)}(z)$ has non-singular OPE\footnote{ Notice that, to prove this, we use that the charge and the conformal weight of $v$ are related by $h(v)=\frac{Q(v)^2}{2}$.} with all currents $j(z)$, $j^+(z)$, and $j^-(z)$ of $\alg{su}(2)_N$, and therefore it is invariant under the $SU(2)$ group of automorphisms. Let $f\in SU(2)$ be an automorphism that flips the sign of $j(z)$ and exchanges $j^+$ and $j^-$. Then, $f$ leaves both summands in the sum $T_V=T_{\alg{u}(1)}+T_\calC$ invariant. As a consequence, $f$  maps the commutant $\Com(T_V(z),\tilde V)\cong \paraf(N)$ of $T_V$ in $\tilde V$ to itself. Similarly, $f$ maps the commutant of $\paraf(N)$, i.e. $\Com( \paraf(N),\tilde V)\cong V^g$ to itself. It follows that $f$ restricts to an automorphism of the form $f_1\otimes f_2$ on the subVOA $\paraf(N)\otimes V^g\subset \tilde V$. Because $f$ exchanges $j^+$ in $\paraf(N,[0,+2])\otimes V_{+ 1,+ 1}$ with  $j^-$ in $\paraf(N,[0,-2])\otimes V_{- 1,- 1}$, $f_1$ must be a charge conjugation automorphism of $\paraf(N)$, that exchanges each module $\paraf(N,[l,m])$ with $\paraf(N,[l,-m])$.\\
Because $f$ is an inner automorphism of $\tilde V$, it must map each module $\tilde V_{\ell,t}$ to itself. Furthermore, since it restricts to an automorphism $f_1\otimes f_2$ of $\paraf(N)\otimes V^g\subset \tilde V$ that acts by charge conjugation on $\paraf(N)$, it must map each component $\paraf(N,[\ell,m])\otimes V_{\frac{m+t}{2},\frac{m-t}{2}}$ to  $\paraf(N,[\ell,-m])\otimes V_{\frac{-m+t}{2},\frac{-m-t}{2}}$. Therefore, $f_2$ must be an outer automorphism of $V^g$ that exchanges each module $V_{a,b}$ with $V_{-b,-a}$. In particular, $f_2$ induces an isomorphism from the original VOA $V=\oplus_{b\in \ZZ/N\ZZ} V_{0,b}$ to its orbifold  $V/\langle g\rangle=\oplus_{b\in \ZZ/N\ZZ} V_{-b,0}$, and we conclude.  \end{proof}

For example, let us consider again the Leech lattice VOA $V_\Lambda$, let $\lambda\in \Lambda$ be a primitive vector with square norm $\lambda^2=2N$, and consider the inner automorphism $h\equiv h_{\frac{\lambda}{N}}\in U(1)^{24}$ acting on a state $\CV_\mu$ with winding-momentum $\mu\in \Lambda$ by
\be h_{\frac{\lambda}{N}}(\CV_\mu)=e^{2\pi i \frac{\lambda\cdot\mu}{N}}\CV_\mu\ .
\ee Then, $h$ has order $N$, and the $h$-twisted sector is $V_{\frac{\lambda}{N}+\Lambda}$, i.e. it includes states with winding momentum in the translate of $\Lambda$ by $\lambda/N$. In particular, there is a single state in $V_{\frac{\lambda}{N}+\Lambda}$ with winding-momentum $\frac{\lambda}{N}$. Its conformal weight is $\frac{1}{2}\frac{\lambda^2}{N^2}=\frac{1}{N}$, so that this is a ground state of the $h$-twisted sector and the symmetry $h$ is non-anomalous. Therefore, theorem \ref{th:invparaf} applies, and $V/\langle h\rangle$ is isomorphic with $V$. 

In general, for a symmetry $h_x\in U(1)^{24}\subset \Aut(V_\Lambda)$ of finite order $N$, with $x\in (\Lambda\otimes\QQ)/\Lambda$, the $h$-twisted sector is $V_{x+\Lambda}$ and its conformal weight is half the square distance between $x$ and the closest lattice vector. The maximal possible square distance of a point $x\in \Lambda\otimes\RR$ from the Leech lattice is $2$, and in this case $x$ is called a `deep hole' -- this corresponds to the $h_x$-twisted sector having the maximal conformal weight $1$. There are $23$ inequivalent deep holes, up to the action by the affine automorphism group $\Lambda\rtimes Co_0$, and the orbifolds of $V_\Lambda$ by such $h_x$ gives rise to the $23$ Niemeier lattice VOAs \cite{MollerScheithauer2023}. The case where theorem \ref{th:invparaf} applies is, in a sense, the opposite extreme, where the conformal weight of the $h_x$-twisted sector $V_{x+\Lambda}$ is as small as possible for a non-anomalous symmetry of order $N$. In these cases, $V_\Lambda$ is always self-orbifold. It would be interesting to derive some general results for the intermediate cases, where the conformal weight is $1-\frac{r}{N}$ with $1<r<N$.

\bigskip

{\bf Acknowledgments.} 
I would like to thank Terry Gannon for many interesting conversations about vertex operator algebras, and Roberta Angius, Sarah Harrison, Alessandro Miccich\`e for useful conversations on related topics. The author would like to thank the Isaac Newton Institute for Mathematical Sciences, Cambridge, for support and hospitality during the programme New connections between physics and number theory, where work on this paper was undertaken. I acknowledge support from CARIPARO Foundation Grant under
grant n. 68079.

\appendix

\section{Vertex operator algebras, modules, and characters}\label{a:VOA}

The chiral algebra of holomorphic fields in a  conformal field theory on the Riemann sphere can be described formally in terms of a vertex operator algebra. In this section, we briefly summarise the basic definitions and our notation; we leave all the details to the references.

A VOA consists of a tuple $(V,Y,1,\omega)$ (denoted  by $V$ for short), where:
\begin{itemize}
	\item $V$ is a $\CC$-vector space, corresponding to the subspace of states in the CFT with $\bar L_0=0$, i.e. that are mapped to holomorphic operators by the state-operator correspondence. It admits a graded decomposition as
	\be V=\oplus_{n\in \ZZ} V_n\ .
	\ee where $\dim V_n<\infty$ for all $n\in \ZZ$, and $V_n=0$ for sufficiently small $n$ ($V_n$ is just the $L_0$-eigenspace with eigenvalue $n$, see below). In this work, we only consider VOAs of CFT type, i.e. with $V_n=0$ for $n<0$, and $V_0\cong \CC$.  They arise as chiral algebras in unitary, compact CFTs with a unique vacuum.
	\item $Y:V\to \End(V)[[z^{\pm 1}]]$ maps every state $v\in V$ to the associated holomorphic vertex operator $Y(v,z)=\sum_{n\in \ZZ} v(n) z^{-n-1}$, seen as a formal power series in $z$ with the coefficients (modes) $v(n):V\to V$ being $\CC$-linear endomorphisms in $V$. Note that the math conventions for the mode grading is different from the physics conventions $Y(v,z)=\sum_{n\in \ZZ} v_n z^{-n-h_v}$, where $h_v$ is the conformal weight ($L_0$-eigenvalue) of $v$.
	\item $1 \in V_0\subset V$ is the non-zero vacuum state, corresponding to the identity vertex operator $Y(1,z)=1$.
	\item $\omega\in V_2\subset V$ is the conformal vector, corresponding to the holomorphic  stress-tensor $T(z):=Y(\omega,z)=\sum_{n\in \ZZ} L_n z^{-n-2}$, where the modes $L_n$ obey the Virasoro algebra for some central charge $c$. One has $L_0v=nv$ for all $v\in V_n$.
\end{itemize}  There are a number of axioms that should be satisfied, that translate the properties expected for the OPEs between holomorphic operators in a CFT, and their covariance with respect to conformal transformations generated by the $L_n$. See \cite{FLM1988} for details.

There are natural notions of homomorphisms, isomorphism, and embeddings between VOAs, and of subVOA $W\subset V$. In the latter case, a subVOA $W$ has the same vacuum as $V$, $1_W=1_V$, but we do not require the conformal vectors $\omega_V$ and $\omega_W$ to coincide.

In a generic 2-dim CFT, the chiral and anti-chiral algebras correspond to a holomorphic and an anti-holomorphic VOAs (not necessarily isomorphic to each other). The Hilbert space decomposes into subspaces associated with a primary state and all its descendants, and can be described in terms of VOA modules (we will also use the word `VOA representations'). We always work with ordinary modules $(M, Y_M)$, where $M$ is a vector space with a decomposition
\be M=\oplus_{\lambda} M_{\lambda}\ ,
\ee into finite dimensional $L_0$-eigenspaces $M_{\lambda}$ and $Y_M:V\to \End(M)[[z^{\pm 1}]]$ maps every VOA vector $v\in V$ to the vertex operator $Y_M(v,z)=\sum_{n\in \ZZ} v_M(n) z^{-n-1}$, where the modes $v_M(n)$ are now linear maps from $M$ to itself. Such vertex operators must obey a number of compatibility axioms, see \cite{FLM1988}. A $V$-module is simple (or irreducible) if it contains no proper submodule. The VOA $V$ is also a module over itself (called the adjoint or vacuum module), and in our examples we always assume that it is simple. 

Given a $V$-module $M$, there is always a contragredient (or dual, or charge conjugate) module, with underlying vector space $M^*=\oplus_{\lambda} M^*_{\lambda}$, where $M^*_{\lambda}=\Hom(M_\lambda,\CC)$ is the dual vector space. A VOA $V$ is self-contragredient if $V\cong V^*$ as a $V$-module. This is always true for the full chiral algebra of a unitary CFT.

In this work, we will only consider rational VOAs, where every ordinary module decomposes as a finite sum of simple modules, and there are only finitely many irreducible modules. A VOA is \emph{strongly rational} if it is rational, simple, self-contragredient, $C_2$-cofinite (i.e. the space $C_2=\{u(-2)v\mid u,v\in V\}$ has finite codimension in $V$, see \cite{Zhu1996ModularInvariance}), and of CFT-type. Ordinary modules of a strongly rational VOA $V$ are objects in a modular tensor category \cite{Huang2008Rigidity}, that we denote by $\mathrm{Rep}(V)$. This category encodes all the usual notions of direct sum, fusion, braiding and the modular data, and in particular the $T$ and $S$ matrices determining the modular transformations of the characters
\be \Tr_{M}(q^{L_0-\frac{c}{24}})\ ,\qquad q=e^{2\pi i\tau}\ .
\ee If $V$ is strongly rational, and $M_1$, $M_2$ are $V$-modules, then we denote by
\be M_1\boxtimes M_2\in \mathrm{Rep}(V)
\ee their fusion product, which is still a $V$-module.

If $V_1$ and $V_2$ are VOAs, and $M_1$ and $M_2$ are, respectively, a $V_1$- and a $V_2$-module, then the tensor product $M_1\otimes M_2$ is a module for the tensor product VOA $V_1\otimes V_2$ (this should not be confused with the fusion product of two modules for the same VOA $V$, considered above).

A VOA is called holomorphic if it is simple, rational, self-contragredient, and the only simple module is the adjoint module, up to isomorphisms.

Let us describe the main examples of VOAs algebras that appear in this work.
\begin{itemize}
	\item Given an even positive definite lattice $L\subset \RR^d$, the lattice VOA $V_L$ is generated by $d$ chiral free bosons $\phi^a$, with winding-momentum lattice $L$. More precisely, $V_L$ contains $d$ commuting  $\alg{u}(1)$ currents $j^a(z)=i\partial \phi^a(z)$, that we usually normalize so that
	\be j^a(z)j^b(0)= \frac{\delta^{ab}}{z^2}+\ldots\ ,
	\ee
	as well as the vertex operators $\mathcal{V}_\lambda(z)\sim  \nord{e^{i\lambda\cdot \phi}}(z)$ of conformal weight $\frac{\lambda^2}{2}$ for all $\lambda\in L$. This VOA is rational, and its irreducible modules are labeled by the cosets in $L^*/L$, where 
	\be L^*=\{\mu\in \RR^d\mid \mu\cdot\lambda\in \ZZ,\ \forall \lambda\in L \}
	\ee is the dual lattice. The characters of the modules $V_{\gamma+L}$, $\gamma\in L^*$, are given by
	\be \Tr_{V_{\gamma+L}}(q^{L_0-\frac{c}{24}})=\frac{\Theta_{\gamma+L}(\tau)}{\eta(\tau)^d}\ ,
	\ee where
	\be \Theta_{\gamma+L}(\tau)=\sum_{\lambda\in \gamma+L} q^{\frac{\lambda^2}{2}}\ ,
	\ee is the theta series of the coset $\gamma+L$. The VOA $V_L$ is holomorphic if and only if the lattice $L$ is unimodular.\\
	For a $1$-dimensional even lattice $L=\sqrt{2k}\ZZ$, we also use the notation $\alg{u}(1)_k\equiv V_{\sqrt{2k}\ZZ}$. Its modules $V_{\frac{m}{\sqrt{2k}}+\sqrt{2k}\ZZ}$ have conformal weight $\frac{m^2}{4k}$, if the representative of $m\in \ZZ/2k\ZZ$ is chosen in the range $\{-k+1,\ldots,k\}$. The (flavoured) characters are
	\be \Tr_{V_{\frac{m}{\sqrt{2k}}+\sqrt{2k}\ZZ}}(q^{L_0-\frac{c}{24}}e^{2\pi i z j_0})=\frac{\Theta^{(k)}_m(\tau,z)}{\eta(\tau)} \ ,
	\ee where
	\be \Theta_m^{(k)}(\tau,z)=\sum_{n\in \ZZ} q^{\frac{1}{4k}(m+2kn)^2}e^{2\pi i z\frac{m+2kn}{2}}\ ,\qquad \eta(\tau)=q^{\frac{1}{24}}\prod_{n=1}^\infty (1-q^n)\ ,
		\ee are the lattice theta series and the Dedekind eta function, respectively. 
	\item Given a simple Lie algebra $g$ and a positive integral level $k\in \ZZ_{> 0}$, the currents of the affine Kac-Moody algebra $\alg{g}_k$ generate a unitary VOA that, by standard abuse of notation, we still denote by $\alg{g}_k$. Its irreducible modules are denoted by $L_\alg{g}(k,\lambda)$, where $\lambda$ is the highest weight of an integrable representation; in particular $\alg{g}_k=L_\alg{g}(k,0)$. More generally, we use the notation $\alg{g}_{1,k_1}\alg{g}_{2,k_2}\cdots \alg{g}_{r,k_r}=L_{\alg{g}_1}(k_1,0)\otimes\ldots\otimes L_{\alg{g}_r}(k_r,0)$. For $\alg{su}(2)_k$, the central charge is
	\be c=\frac{3k}{k+2}\ .
	\ee
	 The modules $L_{\alg{su}(2)}(k,\ell)$ are labeled by $\ell\in \{0,\ldots k\}$, have conformal weight $\frac{\ell(\ell+2)}{4(k+2)}$, and the characters are given by 
	\be\label{su2Nchar} \ch_\ell^{\alg{su}(2)_k}(\tau,z):=\Tr_{L_{\alg{su}(2)}(k,\ell)}(q^{L_0-\frac{c}{24}}e^{2\pi i z j^3_0})=\frac{\Theta^{(k+2)}_{\ell+1}(\tau,z)-\Theta^{(k+2)}_{-\ell-1}(\tau,z)}{\Theta^{(2)}_{1}(\tau,z)-\Theta^{(2)}_{-1}(\tau,z)}\ .\ee Notice that for $\ell,\ell'\in\{0,\ldots, N-1\}$, one has
	\be\label{invertsu2N} \left[\ch_{\ell'}^{\alg{su}(2)_k}(\tau,z)\left( \Theta^{(2)}_{1}(\tau,z)-\Theta^{(2)}_{-1}(\tau,z)\right)\right]_{y^{\ell+1}}=\left[\Theta^{(k+2)}_{\ell'+1}(\tau,z)-\Theta^{(k+2)}_{-\ell'-1}(\tau,z)\right]_{y^{\ell+1}}=q^{\frac{(\ell+1)^2}{4k+8}}\delta_{\ell,\ell'}\ ,
	\ee where $\left[\ldots\right]_{y^{\ell+1}}$ means that we pick the coefficient of the $y^{\ell+1}:=e^{2\pi i z(\ell+1)}$ term.
	\item Given a rational VOA $V$ and a rational subVOA $W\subset V$, the \emph{coset} or \emph{commutant} of $W$ in $V$ is the vertex algebra with vector space
	\be \Com(W,V)=\{v\in V\mid [Y(v,z),Y(w,\zeta)]=0\ \forall w\in W\}\ ,\qquad W\subset V\ ,
	\ee i.e. it is generated by the vertex operators in $V$ that have non-singular OPE with all vertex operators in $W$. While it is not always true that the difference $T_V-T_W$ is contained in $\Com(W,V)$, this will hold in all situations considered in this work (in particular, it always true when $W$ is an affine algebra), so that  $\Com(W,V)$ is a vertex operator algebra with $T_V-T_W$ as a stress tensor. When $V=\alg{g}_k$ and $W=\alg{h}_p$ are both affine algebras, we will use the standard coset notation
	\be \Com(\alg{h}_p,\alg{g}_k)=\frac{\alg{g}_k}{\alg{h}_p}\ .
	\ee
	\item The $\ZZ_k$ parafermion algebra $\paraf(k)$, $k\ge 2$, is the rational VOA obtained as the coset $\paraf(k)\cong \frac{\alg{su}(2)_k}{\alg{u}(1)_k}$, i.e. the commutant $\Com(\alg{u}(1),\alg{su}(2)_k)$ of $\alg{u}(1)_k$ in $\alg{su}(2)_k$. It has central charge
	\be c_{\paraf(k)}=c_{\alg{su}(2)_k}-c_{\alg{u}(1)_k}=\frac{3k}{k+2}-1=\frac{2k-2}{k+2}\ .
	\ee
%
	The modules $\paraf(k,[l,m])$ of $\paraf(k)\cong \paraf(k,[0,0])$ are labeled by pairs $[\ell,m]$, $\ell\in \{0,\ldots,k\}$, $m\in \ZZ/2k\ZZ$, with the condition
	\be \ell\equiv m\mod 2\ ,
	\ee and the field identification
	\be [\ell,m]\equiv [k-l,m\pm k]\ .
	\ee Using this field identification, for every module $\paraf(k,[l,m])=\paraf(k,[k-l,m\pm k])$ one can choose $l$ and $m$ in the sets $l\in \{0,\ldots,k\}$ and $m\in \ZZ$, $-l+1\le m\le l$. With this particular choice, the conformal weight of $\paraf(k,[l,m])$ is given by
	\be h_{l,m}=\frac{l(l+2)}{4(k+2)}-\frac{m^2}{4k}\ .
	\ee
	The characters $\chi^{\paraf(k)}_{[\ell,m]}(\tau)$ of $\paraf(k,[l,m])$ are given by the branching functions
	\be \ch_\ell^{\alg{su}(2)_k}(\tau,z)=\sum_{\substack{m\in \ZZ/2k\ZZ\\ m\equiv \ell\mod 2}} \chi_{[\ell,m]}^{\paraf(k)}(\tau)\frac{\Theta^{(k)}_m(\tau,z)}{\eta(\tau)}\ .
	\ee Explicit formulae for the characters $\chi_{[\ell,m]}^{\paraf(k)}(\tau)$ are given in \cite{Fortin:2006dn}, see also \cite{Haghighat:2023sax}.
\end{itemize}

%

\section{Topological defects}

In this section we introduce some basic notions about topological defects in 2-dimensional CFT; we mostly follow the treatment in \cite{Chang_2019}.

Let $\calC$ denote a unitary (not necessarily holomorphic) conformal field theory on a Euclidean $2$-dimensional space time (worldsheet) with a unique vacuum and compact (i.e. $L_0$ and $\bar L_0$ have discrete spectrum). In this work we mostly consider bosonic CFTs, but generalizations exist for fermionic ones. The only worldsheets we are interested in are the Riemann sphere (possibly without some punctures), the infinite cylinder $S^1\times \RR$, or the torus $S^1\times S^1$.

Besides the usual local operators supported on a point, one can consider correlation functions with the insertions of defects supported on oriented lines. The lines can be either closed or open; in the latter case, one needs to specify defect starting or ending operators at the boundary points. A defect line $\CL$ is called topological if all correlation functions are invariant under continuous deformations of the line supporting $\CL$, as long as one does not crosses the support of some other insertion. We say that a certain local operator $\phi(z,\bar z)$ is preserved by the line defect $\CL$ if every correlation function is invariant when the defect $\CL$ is moved across the insertion of $\phi(z,\bar z)$. One can show that a defect line $\CL$ is topological if and only if both the holomorphic and anti-holomorphic stress-energy tensors $T(z)$ and $\tilde T(\bar z)$ are  preserved by $\CL$. In general, the vector space of operators preserved by a given defect $\CL$ is closed with respect to OPE on the sphere, and in particular the preserved (anti-)holomorphic fields generate a subVOA of the (anti-)chiral algebra of the CFT.

Consider the CFT on a cylinder $S^1\times \RR$, where $S^1$ is the space direction and $\RR$ the Euclidean time direction, and let $\CH$ denote the Hilbert space of states on the circle $S^1$ at fixed time. By the usual state-operator correspondence, $\CH$ can also be identified with the vector space of local operators in the CFT. The are two fundamental ways of inserting a topological defect $\CL$ in such a worldsheet. One possibility is to wrap $\CL$ along the space circle $S^1$ at some fixed time. This associates with every defect $\CL$  a linear operator $\hat\CL:\CH\to \CH$ that commutes with the Virasoro algebra (or, more generally, with the preserved chiral and antichiral algebra). A second possibility is to insert an infinite line defect along the time direction $\RR$ at some fixed space coordinate. This defines a new Hilbert space $\CH_\CL$ (the $\CL$-twisted sector) of $\CL$-twisted states. By the state-operator correspondence, $\CH_\CL$ can be identified with the space of point-like operators starting a defect $\CL$.   The space $\CH_\CL$ is an ordinary module of the preserved chiral and anti-chiral algebras. More generally, the OPE of a local operator in $\CH$ and a $\CL$-twisted operator in $\CH_\CL$ must be again in $\CH_\CL$; thus, $\CH_\CL$ must be a module for the algebra of local operators with respect to the OPE.

The two pieces of data associated with a defect $\CL$, namely the linear operator $\hat\CL$ and the twisted sector $\CH_\CL$, are related by the modular S-transformation on the torus: by defining the $\CL$-twining and the $\CL$-twisted partition functions by
\be Z^\CL(\tau):=\Tr_{\CH}(\hat\CL q^{L_0-\frac{c}{24}}{\bar q}^{\bar L_0-\frac{\bar c}{24}})\ ,\qquad Z_\CL(\tau):=\Tr_{\CH_\CL}(q^{L_0-\frac{c}{24}}{\bar q}^{\bar L_0-\frac{\bar c}{24}})\ ,
\ee one has (for bosonic theories)
\be Z_\CL(\tau)=Z^\CL(-1/\tau)\ .
\ee
This observation puts strong Cardy-like constraints on the possible topological defect lines, in particular if the preserved subalgebra is rational.

Every CFT contains at least one defect, the identity defect $\CI$, that has no effect when it is inserted in a correlation function. The associated linear operator $\hat\CI$ is the identity on $\CH$, and $\CH_\CI\cong \CH$. More generally, if the CFT has a group $G$ of global symmetries, then with each $g\in G$ is associated a topological line defect $\CL_g$. The linear operator $\hat\CL_g$ is just the unitary operator implementing the $g$-transformation on $\CH$, and $\CH_{\CL_g}\equiv \CH_g$ is the $g$-twisted sector. 

There are number of basic operations that are defined on the set of topological defects of a given CFT. First, there is a duality involution $\CL\to \CL^*$ that corresponds to reversing the orientation of the support line; the $\CL^*$-twisted sector $\CH_{\CL^*}$ can be identified with the space of point operators where a line defect $\CL$ terminates. Then, by taking the limit where two parallel defect lines $\CL_1$ and $\CL_2$ are moved very closed to each other, one obtains the fusion $\CL_1\CL_2$. This is an associative, but not necessarily commutative operation, with the defect $\CI$ being the identity. The associated linear operator is just the product $\widehat{\CL_1\CL_2}=\hat\CL_1\hat\CL_2$, while the $\CL_1\CL_2$-twisted sector is a suitable fusion product of modules $\CH_{\CL_1\CL_2}\cong \CH_{\CL_1}\boxtimes\CH_{\CL_2}$. Finally, there is a superposition $\CL_1+\CL_2$ of defects, with associated linear operator $\hat\CL_1+\hat\CL_2$, and with twisted sector being the usual direct sum $\CH_{\CL_1+\CL_2}\cong \CH_{\CL_1}\oplus \CH_{\CL_2}$.

Formally, topological defects in a given CFT can be described as objects in a fusion (or, more generally, tensor) category. The set of morphisms $\Hom(\CL_1,\CL_2)$ between two defects $\CL_1$ and $\CL_2$ is the finite dimensional $\CC$-vector space of topological $2$-way junction point operators, attached to an incoming $\CL_1$ and an outgoing $\CL_2$ defects. More generally, the vector space of $k$-way topological junctions with $r\le k$ incoming defects $\CL_1,\ldots,\CL_r$ and $k-r$ outgoing defects $\CL_{r+1},\ldots,\CL_{k}$ is the vector space $\Hom(\CL_1\cdots\CL_r,\CL_{r+1}\cdots\CL_k)$. A defect $\CL$ is called \emph{simple} if $\Hom(\CL,\CL)\cong \CC$, i.e. if the only topological operators from $\CL$ to itself are proportional to the identity. We always assume that the tensor categories of topological defects in our CFTs are semisimple, i.e. the identity is simple, and every defect decomposes into a finite superposition of simple defects.

The fusion product of two simple defects $\CL_1$ and $\CL_2$ is not necessarily simple, but by semi-simplicity it decomposes as
\be \CL_i\CL_j=\sum_{\text{simple }k} N_{ij}^k\CL_k\ ,
\ee where the fusion coefficients $N_{ij}^k\in \ZZ_{\ge 0}$  equal the dimensions of $3$-junction topological operators
\be N_{ij}^k=\dim_\CC \Hom (\CL_i\CL_j,\CL_k)\ .
\ee If $\CL$ is simple, then also its dual $\CL^*$ is simple, and using $\Hom(\CL,\CL)\cong \Hom(\CL\CL^*,\CI)$, one has 
\be N_{\CL\CL^*}^\CI=1\qquad \forall \text{ simple }\CL\ .
\ee A defect $\CL$ is called \emph{invertible} if it is simple and $\CL\CL^*=\CI$. The set of invertible defects forms a group with respect to fusion, and can be identified with the group of standard global symmetries of the CFT.

The vacuum state of the theory is a simultaneous eigenstate of all linear operators $\hat\CL$ associated with the topological defects $\CL$. The corresponding eigenvalue is called the \emph{quantum dimension} of the defect $\CL$ and denoted by $\langle \CL\rangle$. In a unitary compact CFT with a unique vacuum, $\langle \CL\rangle\ge 1$ for all defects $\CL$, with equality holding if and only if $\CL$ is invertible.

 


 \printbibliography

\end{document}